\documentclass{article}

\usepackage[margin=1in]{geometry}
\usepackage[dvipsnames]{xcolor}
\usepackage{amsmath, amsfonts, amsthm, amssymb}
\usepackage{enumitem}
\setdescription{itemsep=5pt,parsep=0pt,leftmargin=*}
\usepackage{graphicx, caption, subcaption}
\usepackage{nicefrac}
\usepackage{tikz}



\newcommand\cZ{\mathcal{Z}}
\newcommand\cX{\mathcal{X}}
\newcommand\cC{\mathcal{C}}
\newcommand\cF{\mathcal{F}}
\newcommand\cD{\mathcal{D}}
\newcommand\cT{\mathcal{T}}
\newcommand\cG{\mathcal{G}}
\newcommand\cL{\mathcal{L}}

\newcommand\cR{\mathcal{R}}
\newcommand{\reals}{\ensuremath{\mathbb{R}}}
\newcommand{\Econd}[2]{\mathbb{E}\!\left[\, #1 \;\middle\vert\; #2\, \right]}
\DeclareMathOperator*{\argmax}{argmax}	
\DeclareMathOperator*{\argmin}{argmin}	
\DeclareMathOperator{\E}{\mathbb{E}}
\DeclareMathOperator{\Var}{\mathbb{V}}

\newcommand{\ldot}[2]{\ensuremath{\left\langle #1, #2 \right\rangle}}
\newcommand{\cdiv}{\ensuremath{./}}
\newcommand{\cmult}{\ensuremath{\circ}}
\newcommand{\duals}{q}
\newcommand{\interp}{x}
\newcommand{\ubc}{|\cC|}

\newcommand{\Risk}{\mathcal{R}}

\makeatletter
\newcommand\Label[1]{&\refstepcounter{equation}(\theequation)\ltx@label{#1}&}
\makeatother

\newtheorem{claim}{Claim}
\newtheorem{assumption}{Assumption}
\newtheorem{defn}{Definition}
\newtheorem{lemma}{Lemma}
\newtheorem{theorem}{Theorem}

\title{Optimal Data Acquisition for Statistical Estimation}

\author{Yiling Chen\thanks{Harvard University. Email: yiling@seas.harvard.edu.}
\and Nicole Immorlica\thanks{Microsoft Research. Email: nicimm@microsoft.com.}
\and Brendan Lucier\thanks{Microsoft Research. Email: brlucier@microsoft.com.}
\and Vasilis Syrgkanis\thanks{Microsoft Research. Email: vasy@microsoft.com.}
\and Juba Ziani\thanks{California Institute of Technology. Email: jziani@caltech.edu.}
}

\begin{document}

\maketitle

\begin{abstract}
We consider a data analyst's problem of purchasing data from strategic agents to compute an unbiased estimate of a statistic of interest. Agents incur private costs to reveal their data and the costs can be {\em arbitrarily correlated} with their data. Once revealed, data are verifiable.  This paper focuses on {\em linear} unbiased estimators. We design an individually rational and incentive compatible mechanism that optimizes the worst-case mean-squared error of the estimation, where the worst-case is over the unknown correlation between costs and data, subject to a budget constraint in expectation.  We characterize the form of the optimal mechanism in closed-form.  We further extend our results to acquiring data for estimating a parameter in regression analysis, where private costs can correlate with the values of the dependent variable but not with the values of the independent variables.
\end{abstract}

\section{Introduction}

In the age of automation, data is king.  The statistics and machine learning algorithms that help curate our online content, diagnose our diseases, and drive our cars, among other things, are all fueled by data.  Typically, this data is mined by happenstance: as we click around on the internet, seek medical treatment, or drive ``smart'' vehicles, we leave a trail of data.  This data is recorded and used to make estimates and train machine learning algorithms.  So long as representative data is readily abundant, this approach may be sufficient.  But some data is sensitive and therefore inaccurate, rare, or lacking detail in observable data traces.  In such cases, it is more expedient to buy the necessary data directly from the population.

Consider, for example, the problem a public health administration faces in trying to learn the average weight of a population, perhaps as an input to estimating the risk of heart disease.  Weight is a sensitive personal characteristic, and people may be loath to disclose it.  It is also variable over time, and so must be collected close to the time of the average weight estimate in order to be accurate.  Thus, while other characteristics, like height, age, and gender, are fairly accurately recorded in, for example, driver's license databases, weight is not.  The public health administration may try surveying the public to get estimates of the average weight, but these surveys are likely to have low response rates and be biased towards healthier low-weight samples.  

In this paper, we propose a mechanism for buying verifiable data from a population in order to estimate a statistic of interest, such as the expected value of some function of the underlying data.  
We assume each individual has a private cost, or disutility, for revealing his or her sensitive data to the analyst.  Importantly, this cost may be correlated with the private data.  For example, overweight or underweight individuals to have a higher cost of revealing their data than people of a healthy weight.  Individuals wish to maximize their expected utility, which is the expected payment they receive for their data minus their expected cost.  
The analyst has a fixed budget for buying data.  The analyst does not know the distribution of the data: properties of the distribution is what she is trying to learn from the data samples, therefore it is important that she uses the data she collects to learn it rather than using an inaccurate prior distribution (for example, the analyst may have a prior on weight distribution within a population from DMV records or previous surveys, but such a prior may be erroneous if people do not accurately report their weights). However, we do assume the analyst has a prior for the marginal distribution of costs, and that she estimates how much a survey may cost her as a function of said prior.\footnote{This prior could come from similar past exercises. Alternatively, when no prior is known, the analyst can allocate a fraction of his budget to buying data for the sake of learning this distribution of costs.  In this paper, we follow prior work (e.g., Roth and Schoenebeck~\cite{RS12}) and assume that a prior distribution is known, instead of focusing on how one might learn the distribution of costs.}

The analyst would like to buy data subject to her budget, then use that data to obtain an unbiased estimator for the statistic of interest.
To this end, the analyst posts a menu of probability-price pairs.  Each individual $i$ with cost $c$ selects a pair $(A,P)$ from the menu, at which point the analyst buys the data with probability $A$ at price $P$.  The expected utility of the individual is thus $(P-c)A$.\footnote{As we show, this menu-based formulation is fully general and captures arbitrary data-collection mechanisms.}  To form an estimate based on this collected data, we assume the analyst uses {\it inverse propensity scoring}, pioneered by Horvitz and Thompson~\cite{HT52}.  This is the unique unbiased linear estimator; it works by upweighting the data from individual $i$ by the inverse of his/her selection probability, $1/A$.  

The Horvitz-Thompson estimator always generates an unbiased estimate of the statistic being measured, regardless of the price menu.  However, the precision of the estimator, as measured by the variance or mean-squared error of the estimate, depends on the menu of probability-price pairs offered to each individual.  For example, offering a high price would generate data samples with low bias (since many individuals would accept such an offer), but the budget would limit the number of samples.  Offering low prices allows the mechanism to collect more samples, but these would be more heavily biased, requiring more aggressive correction which introduces additional noise.  The goal of the analyst is to strike a balance between these forces and post a menu that minimizes the variance of her estimate in the worst-case over all possible joint distributions of the data and cost consistent with the cost prior.  We note that this problem setting was first studied by Roth and Schonebeck~\cite{RS12}, who characterized an approximately optimal mechanism for moment estimation.

\subsection{Summary of results and techniques} 

Our main contribution comes in the form of an exact solution for the optimal menu, as discussed in Section~\ref{sec:moment-estimation}.  As one would expect, if the budget is large, the optimal menu offers to buy, with probability $1$, all data at a cost equal to the maximum cost in the population.  If the budget is small, 
the optimal menu buys data from an individual with probability inversely proportional to the square root of their cost.\footnote{Of course, the individual is him/herself selecting the menu option and so the use of an active verb in this context is perhaps a bit misleading.  What we mean here is that, given his/her incentives based on his/her private cost, the choice the individual selects is one that buys his/her data with probability inversely proportional to the square root of his/her cost.}  Interestingly, in intermediate regimes, we show the optimal menu employs pooling: for all individuals with sufficiently low private cost, it buys their data with equal probability; for the remaining high cost agents, it buys their data with probability inversely proportional to the square root of their costs.
Revisiting the example of estimating the weight of a population of size $n$, our scheme suggests the following solution.  Imagine the costs are $0,4,8$ with probability $\frac{1}{2},\frac{1}{4},\frac{1}{4}$, and the total budget of the analyst is $B = 7n$.  The analyst brings a scale to a public location and posts the following menu of pairs of allocation probability and price: $\{(1,\frac{36}{5}),(\frac{4}{5},8)\}$.  A simple calculation shows that individuals with cost $0$ or $4$ will pick the first menu option: stepping on the scale and having their weight recorded with probability $1$, and receiving a payment of $\frac{36}{5}$ dollars. Individuals with cost $8$ will pick the second menu option; if they are selected to step on the scale, which happens with probability $\frac{4}{5}$, the analyst records their weight scaled by a factor of $\frac{5}{4}$. This scaling is precisely the upweighting from inverse propensity scoring. In expectation over the population, the analyst spends exactly his budget $7n$.  The estimate is the average of the scaled weights.

We show how to extend our approach in multiple directions.  First, our characterization of the optimal mechanism holds even when the quantity to be estimated is the expected value of a $d$-dimensional moment function of the data.  
Second, we extend our techniques beyond moment estimation to the common task of multi-dimensional linear regression.
In this regression problem, an individual's data includes both features (which are assumed to be insensitive or publicly available) and outcomes (which may be sensitive).  The analyst's goal is to estimate the linear regression coefficients that relate the outcomes to the features.  We make the assumption that an individual's cost is independent of her features, but may be arbitrarily correlated with her outcome.  For example, the goal might be to regress a health outcome (such as severity of a disease) on demographic information.  In this case, we might imagine that an agent incurs no cost for reporting his age, height or gender, but his cost might be highly correlated with his realized health outcome.
%
In such a setting, we show that the asymptotically optimal allocation rule, given a fixed average budget per agent as the number of agent grows large, can be calculated efficiently and exhibits a pooling region as before.  However, unlike for moment estimation, agents with intermediate costs can also be pooled together. We further show that our results extend to non-linear regression in Appendix~\ref{app:nonlinear-regression}, under mild additional conditions on the regression function. 

Our techniques rely on i) reducing the mechanism design problem to an optimization problem through the classical notion of virtual costs, then ii) reducing the problem of optimizing the worst-case variance to that of finding an equilibrium of a zero-sum game between the analyst and an adversary.  The adversary's goal is to pick a distribution of data, conditional on agents' costs, that maximizes the variance of the analyst's estimator.  We then characterize such an equilibrium through the optimality conditions for convex optimization described in~\cite{Boyd04}. 


\subsection{Related work}
A growing amount of attention has been placed on understanding interactions between the strategic nature of data holders and the statistical inference and learning tasks that use data collected from these holders. The work on this topic can be roughly divided into two categories according to whether money is used for incentive alignment. 

In the first category, individuals as data holders do not directly derive utility from the accuracy of the inference or learning outcome, but in some cases may incur a privacy cost if the outcome leaks their private information. The analyst uses monetary payments to incentivize agents to reveal their data.   
Our work falls into this category. Prior papers by Roth and Schoenebeck~\cite{RS12} and Abernethy et al.~\cite{Yiling15} are closest to our setting. Similarly to our work, both Roth and Schoenebeck~\cite{RS12} and Abernethy et al.~\cite{Yiling15} consider an analyst's problem of purchasing data from individuals with private costs subject to a budget constraint, allow the cost to be correlated with the value of data, and assume that individuals cannot fabricate their data. Roth and Schoenebeck~\cite{RS12} aim at obtaining an optimal unbiased estimator with minimum worst-case variance for population mean, while their mechanism achieves optimality only approximately: instead of the actual worst-case variance, a bound on the worst-case variance is minimized. While our setting is identical to that of~\cite{RS12}, our work precisely minimizes worst-case variance (under a regularity assumption on the cost distribution), and our main contribution is to exhibit the structure of the optimal mechanism, as well as to extend our results to broader classes of statistical inference, moment estimation and linear regression. In particular, compared to~\cite{RS12}, our solution exhibits new structure in the form of a pooling region for low cost agents; i.e., the optimal mechanism pools agents with the lowest costs together and treats them identically. Such structure does not arise in~\cite{RS12} under a regularity assumption on the cost distribution. Abernethy at al.~\cite{Yiling15} consider general supervised learning. They do not seek to achieve a notion of optimality; instead, they take a learning-theoretic approach and design mechanisms to obtain learning guarantees (risk bounds). 

Several 
papers consider data acquisition models with different objectives under the assumptions that (a) individuals do not fabricate their data, and (b) private costs and value of data are uncorrelated. For example, in the work of Cummings et al.~\cite{ITCS15}, the analyst can decide the level of accuracy for data purchased from each individual, and wishes to 
guarantee a certain desired level of accuracy of the aggregated information while minimizing the total privacy cost incurred by the agents. Cai et al.~\cite{CDP15} focus on incentivizing individuals to exert effort to obtain high-quality data for the purpose of linear regression. 
Another line of research in the first category examines the data acquisition problem under the lens of differential privacy~\cite{GR11,FL12, GLRS14, NVX14, CIL15}. The mechanism designer then uses payments to balance the trade-off between privacy and accuracy. 


In the second category, individuals' utilities directly depend on the inference or learning outcome (e.g. they want a regression line to be as close to their own data point as possible) and hence they have incentives to manipulate their reported data to influence the outcome. There often is no cost for reporting one's data. The data analyst, without using monetary payments, attempts to design or identify inference or learning processes so that they are robust to potential data manipulations. Most papers in this category assume that independent variables (feature vectors) are unmanipulable public information and dependent variables are manipulable private information~\cite{DFP10,MAMR11,MPR12, PP03}, though some papers consider strategic manipulation of feature vectors~\cite{HMPW16,DRSWW17}. Such strategic data manipulations have been studied for estimation~\cite{CPS16b}, classification~\cite{MAMR11, MPR12,HMPW16}, online classification~\cite{DRSWW17}, regression~\cite{PP04, DFP10}, and clustering~\cite{PP03}. Work in this category is closer to mechanism design without money in the sense that they focus on incentive alignment in acquiring data (e.g., strategy-proof algorithms) but often do not evaluate the performance of the inference or learning, with a few notable exceptions~\cite{HMPW16,DRSWW17}.

\section{Model and Preliminaries}
\paragraph{Survey Mechanisms}
There is a population of $n$ agents. Each agent $i$ has a private pair $(z_i, c_i)$, where $z_i \in \cZ$ is a data point and $c_i > 0$ is a cost. We think of $c_i$ as the disutility agent $i$ incurs by releasing her data $z_i$. The pair is drawn from a distribution $\cD$, unknown to the mechanism designer. We denote with $\cF$ the CDF of the marginal distribution of costs,\footnote{Throughout the text we will use the CDF to refer to the distribution itself.} supported on a set $\cC$. We assume that $\cF$ and the support of the data points, $\cZ$, are known.  However, the joint distribution $\cD$ of data and costs is unknown.

A \emph{survey mechanism} is defined by an allocation rule $A: \cC \rightarrow [0,1]$ and a payment rule $P:\cC \rightarrow \reals$, and works as follows.  Each agent $i$ arrives at the mechanism in sequence and reports a cost $\hat{c}_i$.  The mechanism chooses to buy the agent's data with probability $A(\hat{c}_i)$.  If the mechanism buys the data, then it learns the value of $z_i$ (i.e., agents cannot misreport their data) and pays the agent $P(\hat{c}_i)$.  Otherwise the data point is not learned and no payment is made.

We assume agents have quasi-linear utilities, so that the utility enjoyed by agent $i$ when reporting $\hat{c}_i$ is
\begin{equation}
u(\hat{c}_i; c_i) = \left(P(\hat{c}_i) - c_i\right) \cdot A(\hat{c}_i)
\end{equation}
We will restrict attention to survey mechanisms that are truthful and individually rational.  
\begin{defn}[Truthful and Individually Rational - TIR] A survey mechanism is truthful if for any cost $c$ it is in the agent's best interest to report their true cost, i.e. for any report $\hat{c}$:
\begin{equation}
u(c; c) \geq u(\hat{c}; c)
\end{equation}
It is individually rational if, 
e. for any cost $c\in \cC$, $P(c)\geq c$.
\end{defn}

We assume that the mechanism is constrained in the amount of payment it can make to the agents. We will formally define this as an expected budget constraint for the survey mechanism.
\begin{defn}[Expected Budget Constraint] A mechanism respects a budget constraint $B$ if:
\begin{equation}
n\cdot \E_{c\sim \cF}\left[P(c)\cdot A(c)\right]\leq B
\end{equation} 
\end{defn}

\paragraph{Estimators} The designer (or \emph{data analyst}) wishes to use the survey mechanism to estimate some parameter $\theta \in \reals$ of the marginal distribution of data points.\footnote{We also extend our results to multi-dimensional parameters; see Section~\ref{sec:multidimensional-linear}.}  For example, it might be that $\cZ = \reals$ and $\theta$ is the mean of the distribution over data points in the population.  To this end, the designer will apply an \emph{estimator} to the collection of data points $S$ elicited by the survey mechanism.  We will write $\hat{\theta}_S$ for the estimator used.
Note that the value of the estimator $\hat{\theta}_S$ depends on the sample $S$, but might also depend on the distribution of costs $\cF$ and the survey mechanism.  Due to the randomness inherent in the survey mechanism (both in the choice of data points sampled and the values of those samples), we think of $\hat{\theta}_S$ as a random variable, drawn from a distribution $\cT(\cD, A)$.  We will focus exclusively on \emph{unbiased} estimators.
\begin{defn}[Unbiased Estimator] Given an allocation function $A$, an estimator $\hat{\theta}_S$ for $\theta$ is \emph{unbiased} if for any instantiation of the true distribution $\cD$ its expected value is equal to $\theta$:
\begin{equation}
\E_{\hat{\theta}_S\sim \cT(\cD, A)}\left[\hat{\theta}_S\right] = \theta.
\end{equation}
\end{defn}
Given a fixed choice of estimator, the mechanism designer wants to construct the survey mechanism to minimize the variance (finite sample or asymptotic as the population grows) of that estimator.  Since the designer does not know the distribution $\cD$, we will work with the worst-case variance over all instantiations of $\cD$ that are consistent with the cost marginal $\cF$.

\begin{defn}[Worst-Case Variance]
Given an allocation function $A$ and an instance of the true distribution $\cD$, the \emph{variance} of an estimator $\hat{\theta}_S$ is defined as:
\begin{equation}
\Var(\hat{\theta}_S; \cD, A) =\E_{\hat{\theta}_S \sim \cT(\cD,A)}\left[\left(\hat{\theta}_S - \E\left[\hat{\theta}_S\right]\right)^2\right]
\end{equation}
The \emph{worst-case variance} of $\hat{\theta}_S$ is
\begin{equation}
\Var^*(\hat{\theta}_S; \cF, A) = \sup_{\cD~\mathrm{consistent~with~}\cF} \Var(\hat{\theta}_S;\cD, A).
\end{equation}
\end{defn}

We are now ready to formally define the mechanism design problem faced by the data analyst.
\begin{defn}[Analyst's Mechanism Design Problem]\label{defn:mechanism-design}
Given an estimator $\hat{\theta}_S$ and cost distribution $\cF$,
the goal of the designer is to design an allocation rule $A$ and payment rule $P$ so as to minimize worst-case variance subject to the truthfulness, individual rationality and budget constraints: 
\begin{equation}
\begin{aligned}
\inf_{A,P} ~& \Var^*(\hat{\theta}_S; \cF, A)\\
\mathrm{s.t.}~& n \cdot \E_{c\sim \cF}\left[P(c)\cdot A(c)\right]\leq B\\
~& A,P \mathrm{~define~a~TIR~mechanism} 
\end{aligned}
\end{equation}
\end{defn}

\paragraph{Implementing Surveys as Posted Menus.}
The formulation above describes surveys as direct-revelation mechanisms, where agents report costs.  We note that an equivalent indirect implementation might be more natural: a \emph{posted menu survey} offers each agent a menu of (price, probability) pairs $(p_1, A_1), \dotsc, (p_k, A_k)$.  If the agent chooses $(p_m, A_m)$ then their data is elicited with probability $A_m$, in which case they are paid $p_m$.  Each agent can choose the item that maximizes their expected utility, i.e., $\argmax_{m\in [k]}~(p_m - c)\cdot A_m$. By the well-known \emph{taxation principle}, any survey mechanism can be implemented as a posted menu survey, and the number of menu items required is at most the size of the support of the cost distribution.

\subsection{Reducing Mechanism Design to Optimization}\label{sec: reduction}

We begin by reducing the mechanism design problem to a simpler full-information optimization problem where the designer knows the private cost of each player and can acquire their data by paying them exactly that cost. However, the designer is constrained to using \emph{monotone} allocation rules, in which
players with higher costs have weakly lower probability of being chosen.  

\begin{defn}[Analyst's Optimization Problem]\label{defn:optimization}
Given an estimator $\hat{\theta}_S$ and cost distribution $\cF$, the optimization version of the designer's problem is to find a non-increasing allocation rule $A$ that minimizes worst-case variance subject to the budget constraint, assuming agents are paid their cost:
\begin{equation}
\begin{aligned}
\inf_{A} ~& \Var^*(\hat{\theta}_S; \cF, A)\\
\mathrm{s.t.}~& n \cdot \E_{c\sim \cF}\left[c\cdot A(c)\right]\leq B\\
~& A \mathrm{~is~monotone~non}\text{-}\mathrm{increasing} 
\end{aligned}
\end{equation}
\end{defn}

The mechanism design problem in Definition \ref{defn:mechanism-design} reduces to the optimization problem given by Definition \ref{defn:optimization}, albeit with a transformation of costs to \emph{virtual cost}.
\begin{defn}[Virtual Costs and Regular Distributions]
If $\cF$ is continuous and admits a density $f$ then define the virtual cost function as $\phi(c)=c + \frac{\cF(c)}{f(c)}$. If $\cF$ is discrete with support $\cC=\{c_1, \ldots, c_{|\cC|}\}$ and PDF $f$, then define the virtual cost function as: $\phi(c_t) =  c_t + \frac{c_t-c_{t-1}}{f(c_t)} \cF(c_{t-1})$, with $c_0=0$. We also denote with $\phi(\cF)$ the distribution of virtual costs;
i.e., the distribution created by first drawing $c$ from $\cF$ and then mapping it to $\phi(c)$. A distribution $\cF$ is regular if the virtual cost function is increasing. 
\end{defn}

When $\cF$ is twice-continuously differentiable, $\cF$ is regular if and only if $\cF(c) f'(c) < 2 f(c)^2$ for all $c \in \cC$. Importantly, in this case, the allocation rule of Roth and Schoenebeck~\cite{RS12} is monotone strictly decreasing in $c$ and does not exhibit a pooling region at low-cost as our solution does. The following is an analogue of Myerson's \cite{Myer81} reduction of mechanism design to virtual welfare maximization, adapted to the survey design setting. 

\begin{lemma}\label{lem:reduction-to-opt}
If the distribution of costs $\cF$ is regular, then solving the Analyst's Mechanism Design Problem reduces to solving the Analyst's Optimization Problem for distribution of costs $\phi(\cF)$.
\end{lemma}
\begin{proof}
The proof is given in Appendix~\ref{app:truth-char}
\end{proof}

\subsection{Unbiased Estimation and Inverse Propensity Scoring}

We now describe a class of estimators $\hat{\theta}_S$ that we will focus on for the remainder of the paper.
Note that simply calculating the quantity of interest, $\theta$, on the sampled data points can lead to bias, due to the potential correlation between costs and data.
For instance, suppose that $z\in \reals$ and the goal is to estimate the mean of the distribution of $z$. A natural estimator is the average of the collected data: $\hat{\theta}_S = \frac{1}{|S|} \sum_{i\in S} z_i$. However, if players with lower $z$ tend to have lower cost, and are therefore selected with higher probability by the analyst, 
then this estimator will consistently underestimate the true mean. 


This problem can be addressed using \emph{inverse propensity scoring} (IPS), pioneered by Horvitz and Thompson \cite{HT52}. The idea is to recover unbiasedness by weighting each data point by the inverse of the probability of observing it.  
%
%
%
This IPS approach can be applied to any parameter estimation problem where the parameter of interest is the expected value of an arbitrary \emph{moment function} $m:\cZ\rightarrow \reals$.
\begin{defn}[Horvitz-Thompson Estimator]\label{defn:ht} The Horvitz-Thompson estimator for the case when the parameter of interest is the expected value of a (moment) function $m: \cZ\rightarrow \reals$ is defined as:
\begin{equation}
\hat{\theta}_S = \frac{1}{n} \sum_{i\in [n]} \frac{m(z_i)\cdot 1\{i\in S\}}{A(c_i)}
\end{equation}
\end{defn}
The Horvitz-Thompson estimator is the unique unbiased estimator that is a linear function of the observations $m(z_i)$ \cite{RS12}.  It is therefore without loss of generality to focus on this estimator if one restricts to unbiased linear estimators.\footnote{We note that we have assumed, for convenience, that $A(c_i) > 0$ for all $i \in [n]$ in the expression of this estimator, for it to be unbiased and well-defined.  It is easy to see from the expression for the variance given in Section~\ref{sec:moment-estimation} that the variance-minimizing allocation rule will indeed be non-zero for each cost. }

\paragraph{IPS beyond moment estimation.} We defined the Horvitz-Thompson estimator with respect to moment estimation problems, $\theta = \E[m(z)]$.
As it turns out, this approach to unbiased estimation extends even beyond the moment estimation problem to parameter estimation problems defined as the solution to a system of moment equations $\E[m(z;\theta)]=0$ or parameters defined as the minima of a moment function $\argmin_{\theta}\E[m(z;\theta)]$. We defer this discussion to Section~\ref{sec:linear-regression}.

\section{Estimating Moments of the Data Distribution}\label{sec:moment-estimation}
In this section we 
consider 
the case where the analyst's goal is to estimate the mean of a given moment function of the distribution.  That is, there is some function $m:\cC\rightarrow [0,1]$ 
such that both $0$ and $1$ are in the support of random variable $m(z)$, and the goal of the analyst is to estimate $\theta = \E[m(z)]$.\footnote{Observe that it is easy to deal with the more general case of $m(z) \in [a,b]$ by a simple linear translation, i.e., estimate $\tilde{m}(z) = \frac{m(z)-a}{b-a}$ instead, which is in $[0,1]$ and then translate the estimator back to recover $m(z)$.}  We assume that $\hat{\theta}_S$, the estimator being applied, is the Horvitz-Thompson estimator given in Definition \ref{defn:ht}.

For convenience we will assume that the cost distribution $\cF$ has finite support, say $\cC=\{c_1,\ldots,c_{\ubc}\}$ with $c_1< \ldots < c_{\ubc}$.  (We relax the finite support assumption in Section~\ref{sec:moments.continuum}.)
Write $\pi_t=f(c_t)$ for the probability of cost $c_t$ in $\cF$.  Also, for a given allocation rule $A$, we will write $A_t=A(c_t)$ for convenience.  That is, we can interpret an allocation rule $A$ as a vector of $\ubc$ values $A_1, \dotsc, A_{\ubc}$.   For further convenience, we will write $\duals_t=\Pr[m(z)=1|c_t]$.  This is the probability that the moment takes on its maximum value when the cost is $c_t$.  Finally, we will assume that the distribution of costs is regular.

Our goal is to address the analyst's mechanism design problem for this restricted setting. By Lemma \ref{lem:reduction-to-opt} it suffices to solve the analyst's optimization problem. 
We start by characterizing the worst-case variance for this setting. 

\begin{lemma}\label{lem:ht-variance}
The worst-case variance of the Horvitz-Thompson estimator of a moment $m:\cC\rightarrow [0,1]$, given cost distribution $\cF$ and allocation rule $A$, is:
\begin{equation}
\label{eqn:moment-simplified-variance}
n\cdot \Var^*(\hat{\theta}_S; \cF, A) = \sup_{\duals \in [0,1]^{\ubc}}
\sum_{t=1}^{\ubc} \pi_t \cdot \frac{\duals_t}{A_t} - \left(\sum_{t=1}^{\ubc} \pi_t\cdot \duals_t\right)^2
\end{equation}
\end{lemma}
\begin{proof}
For any distribution $\cD$, observe that the Horvitz-Thompson estimator can be written as the sum of $n$ i.i.d. random variables each with a variance:
\begin{align*}
\sigma^2 = \E\left[\left(\frac{m(z_i)\cdot 1\{i\in S\}}{A(c_i)}\right)^2\right] - \E\left[\frac{m(z_i)\cdot 1\{i\in S\}}{A(c_i)}\right]^2 = \sum_{t=1}^{\ubc} \pi_t  \cdot \frac{\E\left[m(z)^2|c_t\right]}{A_t} - \E[m(z)]^2
\end{align*}
Hence, the variance of the estimator is $\nicefrac{\sigma^2}{n}$. Observe that conditional on any value $c$, the worst-case distribution $\cD$, will assign positive mass only to values $z\in \cZ$ such that $m(z)\in \{0,1\}$.  This is because any other conditional distribution can be altered by a mean-preserving spread, pushing all the mass on these values, while preserving the conditional mean $\E\left[m(z)|c\right]$. This would strictly increase the latter variance. Thus we can assume without loss of generality that $m(z)\in \{0,1\}$, in which case $m(z)^2 = m(z)$ and $\E[m(z)|c]=\Pr[m(z)=1|c]$. Recall that $\duals_t=\Pr[m(z)=1|c_t]$. Then we can simplify the variance as:
\begin{align*}
n\cdot \Var(\hat{\theta}_S; \cD, A) =~ \sum_{t=1}^{\ubc} \pi_t  \cdot \frac{\E[m(z)|c_t]}{A_t} - \E[m(z)]^2=~\sum_{t=1}^{\ubc} \pi_t \cdot \frac{\duals_t}{A_t} - \left(\sum_{t=1}^{\ubc} \pi_t \cdot \duals_t\right)^2
\end{align*}
The theorem follows since the worst-case variance is a supremum over all possible consistent distributions, hence equivalently a supremum over conditional probabilities $\duals:[0,1]^{\ubc}$.
\end{proof}

Given the above characterization of the variance of the estimator, we can greatly simplify the analyst's optimization problem for this setting.  Indeed, it suffices to find the allocation rule $A \in \left(0,1\right]^{\ubc}$ that minimizes \eqref{eqn:moment-simplified-variance}, subject to $A$ being monotone non-decreasing and satisfying the expected budget constraint.

\subsection{Characterization of the Optimal Allocation Rule}

We are now ready to solve the analyst's optimization problem for moment estimation.  In this and all following sections, we denote $\bar{B} = \frac{B}{n}$ for simplicity of notations, and refer to $\bar{B}$ as the ``average budget per agent''. Note that different agents with different costs may be allocated different fractions of the total budget $B$ that in general do not coincide with $\bar{B}$. We remark that if $\bar{B}$ is larger than the expected cost of an agent, then it is feasible (and hence optimal) for the analyst to set the allocation rule to pick any type with probability $1$. We therefore assume without loss of generality that $\E[c] > \bar{B}$. 

Our analysis is based on an equilibrium characterization, where we view the analyst choosing $A$ and the adversary choosing $z$ as playing a zero-sum game and solve for its equilibria. We first present the characterization and some qualitative implications and then present an outline of our proof. We defer the full details of the proof to Appendix~\ref{app:proof-of-moment-estimation}.

\begin{theorem}[Optimal Allocation for Moment Estimation]\label{thm:main-moment-estimation}
The optimal allocation rule $A$ is determined by two constants $\bar{A}$ and $t^*\in \{0,\ldots, \ubc\}$ such that:
\begin{equation}
A_t = \begin{cases}
\bar{A} & \text{if $t\leq t^*$}\\
\frac{\alpha}{\sqrt{c_t}} & \text{o.w.}
\end{cases}
\end{equation}
with $\alpha$ uniquely determined such that the budget constraint is binding.\footnote{The explicit form of this is $\alpha=\frac{\bar{B} - \bar{A} \E[c\cdot 1\{c\leq c_{t^*}\}]}{\E[\sqrt{c}\cdot 1\{c>c_{t^*}\}}$.} 
Moreover, the parameters $\bar{A}$ and $t^*$ can be computed in time $O(\log(|\cC|))$.  
\end{theorem}

The parameters $\bar{A}$ and $t^*$ in Theorem~\ref{thm:main-moment-estimation} are explicitly derived in closed form in Appendix \ref{app:proof-of-moment-estimation}.  For instance, when $\bar{B}\geq \frac{c_{|\cC|}}{2}$, then $t^*=|\cC|$ and $A_t=\bar{A}=\min\left\{1, \nicefrac{\bar{B}}{\E[c]}\right\}$ for all $t$.  When $\bar{B}\leq \frac{\sqrt{c_1}\E[\sqrt{c}]}{2}$ then $t^*=0$ and $A_t = \frac{\bar{B}}{\sqrt{c_t} \E[\sqrt{c}]}$. In fact, it can be shown (see full proof) that in this latter case, the worst-case distribution is given by $q=1$.  In particular, in this restricted case, the approximation of Roth and Schoenebeck~\cite{RS12} is in fact optimal, and indeed our allocation rule is expressing the solution of Roth and Schoenebeck~\cite{RS12} as a posted menu for a discrete, regular distribution of costs. In every other case, $q \neq 1$ and our solution differs from that of~Roth and Schoenebeck~\cite{RS12}, exhibiting a pooling region for low-cost agents. More generally, the computational part of Theorem~\ref{thm:main-moment-estimation} follows by performing binary search over the support of $\cF$, which can be done in $O(\log(|\cC|))$ time.

\begin{figure}[htpb]
\centering
\begin{subfigure}[t]{0.32\textwidth}
\centering
\includegraphics[scale=.55]{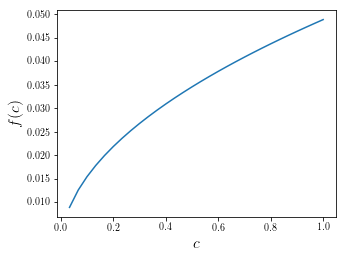}
\end{subfigure}
~~~~~~~~~~~~~
\begin{subfigure}[t]{0.45\textwidth}
\centering
\includegraphics[scale=.551]{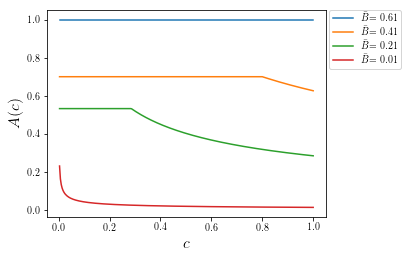}
\end{subfigure}
\caption{The pdf (left) of a distribution of costs
and the corresponding optimal allocation rule for varying levels of per-agent budget (right). 
Note that for sufficiently large budgets, a flat pooling region forms for agents with low costs.
}
\label{fig:example_alloc}
\end{figure}

We note that the optimal rule essentially allocates to each agent inversely proportionally to the square root of their cost, but may also ``pool'' the allocation probability for agents at the lower end of the cost distribution.  See Figure \ref{fig:example_alloc} for examples of optimal solutions.  

The proof of Theorem~\ref{thm:main-moment-estimation} appears in Appendix \ref{app:proof-of-moment-estimation}.  The main idea is to view the optimization problem as a zero-sum game between the analyst who designs the allocation rule $A$, and an adversary who designs $\duals$ so as to maximize the variance of the estimate.  We show how to compute an equilibrium $(A^*,z^*)$ of this zero-sum game via Lagrangian and KKT conditions, and then note that the obtained $A^*$ must in fact be an optimal allocation rule for worst-case variance.


The analysis above applied to a discrete cost distribution over a finite support of possible costs.  We show how to extend this analysis to a continuous distribution below, noting that the continuous variant of the Optimization Problem for Moment Estimation can be derived by taking the limit over finer and finer discrete approximations of the cost distribution.


\subsection{Continuous Costs for Moment Estimation}\label{sec:moments.continuum}
\begin{defn} [Continuous Optimization Problem for Moment Estimation] When costs are supported on $\cC = [0,1]$, the analyst's optimization problem for the moment estimation problem based on the Horvitz-Thompson estimator can be written as:
\begin{equation}\label{eqn:continuous}
\begin{aligned}
\inf_{A: \left(0,1\right]\rightarrow [0,1]}\sup_{\interp: [0,1]\rightarrow [0,1]}~&
\int_{0}^{1}  \frac{\interp(c)}{A(c)}~ d\cF(c) - \left(\int_{0}^1 \interp(c)~d\cF(c)\right)^2\\
\mathrm{s.t.}~& \int_{0}^{1} c \cdot A(c)~d\cF(c) \leq \bar{B}\\
~& A \mathrm{~is~monotone~non}\text{-}\mathrm{increasing} 
\end{aligned}
\end{equation}
\end{defn}
We can now establish the following continuous variant of Theorem~\ref{thm:main-moment-estimation}, which describes the optimal survey mechanism for continuous cost distributions.  
\begin{theorem}[Continuous Limit of Optimal Allocation]\label{thm:main-moment-estimation-continuous}
If the distribution of costs is atomless and supported in $(0,1]$, then the optimal allocation rule $A$ is determined by two constants $\bar{A}$ and $x^*\in \reals$ such that:
\begin{equation}
A(c) = \begin{cases}
\bar{A} & \text{if $c\leq x^*$}\\
\frac{\alpha}{\sqrt{c}} & \text{o.w.}
\end{cases}
\end{equation}
with $\alpha$ uniquely determined such that the budget constraint is binding.\footnote{The explicit form of this is $\alpha=\frac{\bar{B} - \bar{A} \E[c\cdot 1\{c\leq x*\}]}{\E[\sqrt{c}\cdot 1\{c>x^*\}]}$.}. The quantities $\bar{A}$ and $x^*$ are defined as follows: for any $x\in \reals$ let
\begin{align*}
Q_{\infty}(x) =~& \E_{c \sim \cF}[\min\{c, \sqrt{cx}\}]  &
R_{\infty}(x) =~& 2\E_{c \sim \cF}\left[\min\left\{\frac{c}{x}, 1\right\}\right]  &
G(x) =~& \frac{Q_{\infty}(x)}{\max(1, R_{\infty}(x))} 
\end{align*}
Then $x^* = \min\{ 1, G^{-1}(\bar{B})\}$ and $\bar{A}=\frac{1}{\max(1, R_{\infty}(x^*))}$ (see Figure \ref{fig:continuous}).\footnote{We take the convention that if $\bar{B}$ lies above the range of $G$, then $G^{-1}(\bar{B}) = +\infty$.}
\end{theorem}

\begin{proof}
See Appendix~\ref{app:moments.continuum}.
\end{proof}

Let us give some intuition behind the form of the allocation rule described in Theorem~\ref{thm:main-moment-estimation-continuous}.  As in Theorem~\ref{thm:main-moment-estimation}, the allocation rule will pool agents with low costs (i.e., less than some threshold $x^*$), then allocate to higher-cost agents inversely proportional to the square root of their costs.  In the definition of $x^*$ and $\bar{A}$, note that $Q_{\infty}$ is non-decreasing and $R_{\infty}$ is non-increasing, so $G$ is non-decreasing.  We therefore have that $x^*$, the boundary of the pooling region, increases with $\bar{B}$ up to a maximum value of $1$ (at which point all agents are pooled).

Let's restrict attention to the case where the mean of the distribution is at least as large as half of the maximum value of the support, i.e. $\E[c]\geq \nicefrac{1}{2}$.  In this setting, we see that $R_{\infty}(x)\leq 1$ for all $x\in [0,1]$, so
\begin{equation}
G(x) = \frac{Q_{\infty}(x)}{R_{\infty}(x)} = \frac{x}{2}\cdot \frac{\E_{c \sim \cF}[\min\{c, \sqrt{cx}\}]}{\E_{c \sim \cF}\left[\min\left\{c, x\right\}\right]} \approx \frac{x}{2} \tag{see Figure \ref{fig:continuous}}
\end{equation}
So the optimal allocation sets $x^* \approx 2\bar{B}$.  
Moreover, the allocation for the pooling region is $\bar{A} = \frac{1}{R_{\infty}(x^*)} \approx \frac{\bar{B}}{ \E[\min\{c, 2\bar{B}\}]}$.  So the optimal mechanism takes the following intuitive form: first, assign each agent an allocation probability $\bar{A}$ that would, in an alternate world where costs are capped at $2\bar{B}$, precisely exhaust the budget.  Since costs can actually be greater than $2\bar{B}$, this flat allocation goes over-budget.  So, for agents whose costs are greater than $2\bar{B}$, we remove allocation probability so that (a) the budget becomes balanced, and (b) the remaining probability of allocation is inversely proportional to the square root of the costs.



%
\begin{figure}
\centering
\begin{picture}(200,100)
\put(0,0){\includegraphics[scale=.57]{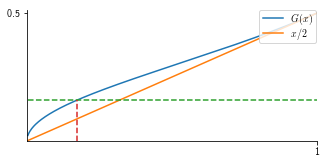}}
\put(187,5){$c$}
\put(1,33){$\bar{B}$}
\put(40,1){$x^*$}
\end{picture}
\caption{Pictorial representation of the function $G(x)$ which defines the optimal threshold $x^*$, portraying its close approximation by the function $x/2$.}\label{fig:continuous}
\end{figure}
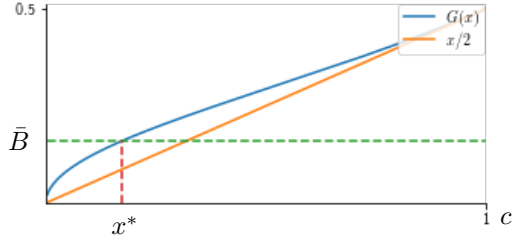

\section{Multi-dimensional Parameters for moment estimation}\label{sec:multidimensional-linear}

Section~\ref{sec:moment-estimation} focused on the case of estimating a single-dimensional parameter of the data distribution.  In this section we note that our characterization of the optimal mechanism extends to multi-dimensional moment estimation as well.  In multi-dimensional moment estimation, there is a function $m \colon \cC \to [0,1]^d$, and our goal is to estimate $\theta = \E[m(z)] \in [0,1]^d$.  Here $d$ is the dimension of the estimation problem, which we assume to be a fixed constant.

As before, we will estimate $\theta$ by applying an estimator $\hat{\theta}$ to the data collected from a survey mechanism.  To evaluate an estimator, we must extend our definition of variance to the $d$-dimensional setting, as follows.
\begin{defn}[Worst-Case Mean Squared Error - Risk]
Given allocation function $A$ and distribution $\cD$, the expected mean squared error (or \emph{risk}) of an estimator $\hat{\theta}$ is
\begin{equation}
\Risk(\hat{\theta}_S; \cD, A) =\E_{\hat{\theta}_S \sim \cT(\cD,A)}\left[\left\|\hat{\theta}_S - \theta_0\right\|_2^2\right]
\end{equation}
and the worst-case variance of $\hat{\theta}$ is
\begin{equation}
\Risk^*(\hat{\theta}_S; \cF, A) = \sup_{\cD~\mathrm{consistent~with~}\cF} \Risk(\hat{\theta}_S;\cD, A).
\end{equation}
\end{defn}

When $\hat{\theta}$ is unbiased, the risk has a natural interpretation: it is simply the sum of variances of each coordinate of $\theta$, considered separately.

\begin{claim}[Risk of Unbiased Estimators]
The risk of any unbiased estimator is equal to the sum of variances of every coordinate:
\begin{equation}
\Risk(\hat{\theta}_S; \cD, A) = \E_{\hat{\theta}_S \sim \cT(\cD,A)}\left[\sum_{r=1}^d (\hat{\theta}_{S,r} -\E[\hat{\theta}_{S,r}])^2\right] = \sum_{r=1}^d \Var(\hat{\theta}_{S,r}).
\end{equation}
\end{claim}

As in the single-dimensional case, the analyst obtains an estimate through the Horvitz-Thompson estimator, which is defined as follows for parameters in $\reals^d$.  Also as in the single-dimensional case, The Horvitz-Thompson estimator is an unbiased estimator of $\E[m(z)]$.
\begin{defn}[Horvitz-Thompson Estimator for Multi-dimensional Moment Estimation]\label{defn:ht.multi} The Horvitz-Thompson estimator for the case when the parameter of interest is the expected value of a vector of moments $m: \cZ\rightarrow \reals^d$ is defined as:
\begin{equation}
\hat{\theta}_S = \frac{1}{n} \sum_{i\in [n]} \frac{1\{i\in S\}}{A(c_i)}\cdot m(z_i)
\end{equation}
\end{defn}

For our characterization of worst-case risk, we will assume that the moment function $m$ can take on the extreme points of the hypercube $[0,1]^d$.  
\begin{assumption}\label{as:support-hypercube}
$\cD$ is such that the induced distribution of $m(z)$ is supported on every extreme point of the $[0,1]^d$ hypercube.
\end{assumption}

\begin{lemma}\label{lem:ht-variance.multi}
Under Assumption~\ref{as:support-hypercube}, the worst-case risk of the Horvitz-Thompson estimator of moment $m:\cC\rightarrow [0,1]$ is
\begin{equation}
\frac{n}{d}\cdot \Risk^*(\hat{\theta}_S; \cF, A) = \sup_{\duals \in [0,1]^{\ubc}}
\sum_{t=1}^{\ubc} \pi_t \cdot \frac{\duals_t}{A_t} - \left(\sum_{t=1}^{\ubc} \pi_t\cdot \duals_t\right)^2.
\end{equation}
\end{lemma}


\begin{proof}
See Appendix~\ref{app:ht-variance}.
\end{proof}

Lemma~\ref{lem:ht-variance.multi} implies that the optimal survey design problem in the $d$-dimensional case is, in fact, identical to the problem considered in the single-dimensional case.  We can conclude that Theorems~\ref{thm:main-moment-estimation} and~\ref{thm:main-moment-estimation-continuous}, which characterized the optimal survey mechanisms for discrete and continuous single-parameter settings, respectively, also apply to the multi-dimensional setting without change.  

%
%

\section{Multi-dimensional Parameter Estimation via Linear Regression}\label{sec:linear-regression}

In this section, we extend beyond moment estimation to a multi-dimensional linear regression task (we discuss the non-linear case in Appendix~\ref{app:nonlinear-regression}).  For this setting we will impose additional structure on the data held by each agent.  Each agent's private information consists of a feature vector $x_i \in \reals$, an outcome value $y_i \in \reals$, and a residual value $\epsilon_i \in \reals$, that are i.i.d among agents.  Each agent also has a cost $c_i$.  The data is generated in the following way: first, $x_i$ is drawn from an unknown distribution $\cX$.  Then, independently from $x_i$, the pair $(c_i, \epsilon_i)$ is drawn from a joint distribution $\cD$ over $\reals^2$.  The marginal distribution over costs, $\cD_c$, is known to the designer, but not the full joint distribution $\cD$.  Then $y_i$ is defined to be
\begin{equation}
    y_i = x_i^\top \theta^* + \epsilon_i
\end{equation}
where $\theta^* \in \Theta$ with $\Theta$ a compact subset of $\reals^d$. We further require that $\theta^*$ is in the interior of $\Theta$.  We write $\cD_\epsilon$ for the marginal distribution over $\epsilon_i$, which is supported on some bounded range $[L,U]$ and has mean $0$. (In particular, $L \leq 0 \leq U$.) We remark that it may be the case, however, that $\E \left[\epsilon_i | c_i \right] \neq 0$. 

When a survey mechanism buys data from agent $i$, the pair $(x_i, y_i)$ is revealed.  Crucially, the value of $\epsilon_i$ is not revealed to the survey mechanism.  The goal of the designer is to estimate the parameter vector $\theta^*$.

Note that the single-dimensional moment estimation problem from Section~\ref{sec:moment-estimation} is a special case of linear regression.  Indeed, consider setting $d = 1$, $\epsilon_i = m(z_i) - \E[m(z_i)]$ for each $i$, $\theta^* = \E[m(z_i)]$, and $x_i$ to be the constant $-1$.  Then, when the survey mechanism purchases data from agent $i$, it learns $y_i = m(z_i)$, and estimating $\theta^*$ is equivalent to estimating the expected value of $m(z_i)$.

More generally, one can interpret $x_i$ as a vector of publicly-verifiable information about agent $i$, which might influence a (possibly sensitive) outcome $y_i$.  For example, $x_i$ might consist of demographic information, and $y_i$ might indicate the severity of a medical condition.  The coefficient vector $\theta^*$ describes the average effect of each feature on the outcome, over the entire population.  Under this interpretation, $\epsilon_i$ is the residual agent-specific component of the outcome, beyond what can be accounted for by the agent's features.  We can interpret the independence of $x_i$ from $(c_i, \epsilon_i)$ as meaning that each agent's cost to reveal information is potentially correlated with their (private) residual data, but is independent of the agent's features.  


As in Section~\ref{sec:moment-estimation}, the analyst wants to design a survey mechanism to buy from the agents,  obtain data from the set $S$ of elicited agents, then compute an estimate $\hat{\theta}_S$ of $\theta$. The expected average payment to each of the n agents should be no more than $\bar{B}$. As in Section~\ref{sec: reduction}, we note that the problem of designing a survey mechanism in fact reduces to that of designing an allocation rule $A$ that minimizes said variance and satisfies a budget constraint in which the prices are replaced by known virtual costs.  To this end, the analyst designs an allocation rule $A$ and a pricing rule $P$ so as to minimize the $\sqrt{n}$-normalized worst-case asymptotic mean-squared error of $\hat{\theta}_S$ as the population size goes to infinity. Our mechanism will essentially be optimizing the coefficient in front of the leading $1/n$ term in the mean squared error, ignoring potential finite sample deviations that decay at a faster rate than $1/n$. Note that we will design allocation and pricing rules to be independent of the population size $n$; hence, the analyst can use the designed mechanism even if the exact population size in unknown.

\subsection{Estimators for Regression}

Let $S$ be the set of data points elicited by a survey mechanism. The analyst's estimate will then be the value $\hat{\theta}_S$ that minimizes the Horvitz-Thompson mean-squared error $\E[(y_i-x_i^\top \theta^*)^2]$, i.e.,
\begin{equation}
 \hat{\theta}_S = \argmin_{\theta \in \Theta} \sum_i  \frac{1 \{i \in S\}}{A(c_i)} (y_i-x_i^\top \theta)^2.
\end{equation}
Further, we make the following assumptions on the distribution of data points:
\begin{assumption}[Assumption on the distribution of features]\label{as: x-distribution}
$E[x_i x_i^\top]$ is finite and positive-definite, and hence invertible.
\end{assumption}
Finite expectation is a property one may expect real data such has age, height, weight, etc. to exhibit. The second part of the assumption is satisfied by common classes of distributions, such as multivariate normals. We first show that $\hat{\theta}_S$ is a consistent estimator of $\theta$.
\begin{lemma}\label{lem: linear-consistent-estimator}
Under Assumption~\ref{as: x-distribution}, for any allocation rule $A > 0$ that does not depend on $n$, $\hat{\theta}_S$ is a consistent estimator of $\theta^*$.
\end{lemma}

\begin{proof}[Proof of Lemma~\ref{lem: linear-consistent-estimator}]
Let $m(\theta;x,y) = (y-x^\top \theta)^2$, and let $w_i = 1\{ i \in S \}$ for simplicity.  The following holds:
\begin{enumerate}
\item First, we note that $\theta^*$ is the unique parameter that minimizes $\E [(y_i-x_i^\top \theta)^2]$; indeed, take any $\theta \neq \theta^*$, we have that 
\begin{align*}
\E [(y_i-\theta^\top x_i)^2] &= \E \left[ \left(y_i-x_i^\top \theta^{*} + x_i^\top (\theta^*-\theta))^2 \right) \right] 
\\&= \E \left[ \left(y_i-x_i^\top \theta^{*} \right)^2 \right] + \E \left[ \left( x_i^\top (\theta^*-\theta) \right)^2  \right] + 2 \E \left[\epsilon_i  (\theta^*-\theta)^\top x_i \right] 
\end{align*}
As $x$ and $\varepsilon$ are independent, $\varepsilon$ has mean $0$, this simplifies to
\begin{align*}
\E [(y_i-\theta^\top x_i)^2] 
&= \E \left[ \left (y_i-\theta^{*\top} x_i\right)^2 \right] + (\theta^*-\theta)^\top \E \left[  x_i x_i^\top \right] (\theta^*-\theta) + 2 (\theta^*-\theta)^\top \E \left[\epsilon_i x_i \right] 
\\&= \E \left[ \left (y_i-\theta^{*\top} x_i \right)^2 \right] + (\theta^*-\theta)^\top \E \left[  x_i x_i^\top \right] (\theta^*-\theta) 
\\&>  \E \left[ \left (y_i-\theta^{*\top} x_i \right)^2 \right]
\end{align*}
where the last step follows from  $\E[ x_i x_i^\top]$ being positive-definite by Assumption~\ref{as: x-distribution}. 
\item By definition, $\Theta$ is compact.
\item $m(\theta;x,y)$ is continuous in $\theta$, and so is its expectation.
\item $m(.;.)$ is also bounded (lower-bounded by $0$, and upper-bounded by either $L^2$ or $U^2$), implying that $\theta \to \frac{w_i }{A(c_i)} m(\theta;x_i,y_i)$ is continuous and bounded. Hence, by the uniform law of large number, remembering that $\frac{w_i }{A(c_i)} m(\theta;x_i,y_i) $ are i.i.d,  
$$\sup_{\theta \in \Theta} \left| \frac{1}{n} \sum_{i=1}^n \frac{w_i }{A(c_i)} m(\theta;x_i,y_i) - \E\left[\frac{w_i }{A(c_i)} m(\theta;x_i,y_i) \right] \right| \to 0.$$
Finally, noting that conditional on $c_i$, $m(\theta;x_i,y_i)$  and $\frac{w_i }{A(c_i)}$ are independent, we have:
\begin{align*}
\E\left[\frac{w_i }{A(c_i)} m(\theta;x_i,y_i) \right] = \E \left[ \Econd{\frac{w_i}{A(c_i)}}{c_i} \Econd{m(\theta;x_i,y_i)}{c_i} \right] 
= \E\left[ m(\theta;x_i,y_i) \right] 
\end{align*}
using $\Econd{\frac{w_i}{A(c_i)}}{c_i} = 1$.
\end{enumerate}
Therefore, all of the conditions of Theorem 2.1 of~\cite{McFadden86} are satisfied, which is enough to prove the result. 
\end{proof}

Similarly to the moment estimation problem in Section~\ref{sec:moment-estimation}, the goal of the analyst is to minimize the worst-case (over the distribution of data and the correlation between $c_i$'s and $\epsilon_i$'s) asymptotic mean-squared error of the estimator $\hat{\theta}_S$. Here ``asymptotic'' means the worst-case error as $\hat{\theta_S}$ approaches the true parameter $\theta^*$.  The following theorem characterizes the asymptotic covariance matrix of $\hat{\theta}_S$.  (In fact, it fully characterizes the asymptotic distribution of $\hat{\theta}_S$.)
\begin{lemma}\label{lem: linear-asymptotic-variance}
Under Assumption~\ref{as: x-distribution}, for any allocation rule $A > 0$ that does not depend on $n$, the asymptotic distribution of $\hat{\theta}_S$ is given by 
\begin{align*}
\sqrt{n} (\hat{\theta}_S - \theta^*) \xrightarrow[]{d} \mathcal{N}\left(0,  \E[x_i x_i^\top]^{- 1}  \E \left[ \epsilon_i^2 \frac{1 \{i \in S \} }{ A^2(c_i) } \right]   \right)
\end{align*}
where $d$ denotes convergence in distribution and where randomness in the expectations is taken on the costs $c_i$, the set of elicited data points $S$, the features of the data $x_i$, and the noise $\epsilon_i$.
\end{lemma}

\begin{proof}[Proof of Lemma~\ref{lem: linear-asymptotic-variance}]
For simplicity, let $w_i = 1\{ i \in S\}$ and note that the $w_i$'s are i.i.d. Let $m(\theta;x_i,y_i) = (y_i-x_i^\top \theta)^2$. First we remark that $\nabla_\theta m(\theta;x_i,y_i) \cdot \frac{w_i}{A(c_i)}  = 2 \frac{w_i}{A(c_i)} x_i(x_i^\top \theta-y_i)$  and $\nabla^2_{\theta \theta} m(\theta;x_i,y_i) \cdot \frac{w_i}{A(c_i)}  = 2 \frac{w_i}{A(c_i)} x_i x_i^\top$. We then note the following:
\begin{enumerate}
\item $\theta^*$ is in the interior of $\Theta$.
\item $\theta \to m(\theta;x_i,y_i) \cdot \frac{w_i}{A(c_i)} $ is twice continuously differentiable for all $x_i,y_i,c_i,w_i$.
\item $\sqrt{n} \left( \frac{1}{n}  \sum_{i=1}^n \nabla_\theta~m(\theta^*;x_i,y_i) \cdot \frac{w_i}{A(c_i)}  \right) \to \mathcal{N}\left(0,4 \E \left[\frac{w_i}{A^2(c_i)} x_i x_i^\top (x_i^\top \theta^* - y_i)^2 \right] \right) $. This follows directly from applying the multivariate central limit theorem, noting that 
$$
\E \left[ \nabla_\theta m(\theta^*;x_i,y_i) \cdot \frac{w_i}{A(c_i)}  \right] = \E \left[ \Econd{2 x_i \epsilon_i}{c_i} \cdot \Econd{\frac{w_i}{A(c_i)}}{ c_i} \right]  =  \E\left[2 x_i \epsilon_i \right] = 0 
$$
where the first step follows from conditional indepence on $c$ of $x,\epsilon$ with $A(c), S$, the second step from $ \Econd{\frac{w_i}{A(c_i)}}{ c_i } = 1$, and the last equality follows from the fact that $x$ and $\epsilon$ are independent and $\E[\epsilon_i ] = 0$.
\item $\sup_{\theta \in \Theta} \left \Vert \E \left[ \nabla_{\theta \theta}^2 m(\theta;x_i,y_i) \cdot \frac{w_i}{A(c_i)}  \right] - \frac{1}{n} \sum_i \nabla_{\theta \theta}^2 m(\theta;x_i,y_i) \cdot \frac{w_i}{A(c_i)}  \right \Vert \to 0$, applying the uniform law of large numbers as $\nabla_{\theta \theta}^2 m(\theta;x_i,y_i) \cdot \frac{w_i}{A(c_i)}  = 2 \frac{w_i}{A(c_i)} x_i x_i^\top$ is i) continuous in $\theta$, and ii) constant in $\theta$, thus bounded coordinate-by-coordinate by $2 \frac{w_i}{A(c_i)} x_i x_i^\top$ that is independent of $\theta$ and has finite expectation $2\E[x_i x_i^\top]$.
\item $\E \left[ \nabla_{\theta \theta}^2 m(\theta;x_i,y_i) \cdot \frac{w_i}{A(c_i)}\right] = 2\E[x_i x_i^\top]$ is invertible as it is positive-definite. 
\end{enumerate}
Therefore the sufficient conditions i)-v) in Theorem 3.1 of~\cite{McFadden86} hold, proving that the asymptotic distribution is normal with mean $0$ and variance 
\[
\E[2x_i x_i^\top]^{- 1}  \E \left[4 \frac{w_i }{ A^2(c_i) } (x_i^\top \theta - y_i)^2 x_i x_i^\top \right]\E[2x_i x_i^\top]^{- 1}.
\] 
To conclude the proof, we remark that by independence of $x_i$ with $c_i$ and $\varepsilon_i$,
\begin{align*}
 \E \left[\frac{w_i }{ A^2(c_i) } (x_i^\top \theta - y_i)^2 x_i x_i^\top \right] = \E [x_i x_i^\top]   \E \left[ \epsilon_i^2 \frac{w_i }{ A^2(c_i) } \right]
\end{align*}
\end{proof}
Lemma~\ref{lem: linear-asymptotic-variance} implies that the worst-case asymptotic mean-squared error, under a budget constraint, is given by the worst-case trace of the variance matrix.  That is,
\begin{equation}
\begin{aligned}
\cR^*(\cF, A) \triangleq  \sup_{\cX}~&\sup_{\cD_\epsilon}~ \E \left[ \epsilon_i^2 \frac{1 \{i \in S \} }{ A^2(c_i) } \right] \cdot \sum_{j=1}^d \E[x_i x_i^\top]_{jj}^{- 1}   
\\&\text{s.t.}~\E[\epsilon_i] = 0
\end{aligned}
\end{equation}
where recall that $\cD_\epsilon$ is the marginal distribution over $\epsilon$ and $\cX$ the distribution over $x$. Importantly, this can be rewritten as
\begin{equation}
\begin{aligned}
\cR^*(\cF, A) \triangleq  \left( \sup_{\cX} \sum_{j=1}^d \E[x_i x_i^\top]_{jj}^{- 1} \right) \cdot~&\sup_{\cD_\epsilon}~\E \left[ \epsilon_i^2 \frac{1 \{i \in S \} }{ A^2(c_i) } \right]
\\&\text{s.t.}~\E[\epsilon_i] = 0   
\end{aligned}
\end{equation}
Therefore, the analyst's decision solely depend on the worst-case correlation between costs $c_i$ and noise $\epsilon_i$, and not on the worst-case distribution $\cX$. In turn, the analyst's allocation is completely independent of and robust in $\cX$.

\subsection{Characterizing the Optimal Allocation Rule for Regression}

As in Section~\ref{sec:moment-estimation}, we assume costs are drawn from a discrete set, say $\cC = \{c_1, \ldots, c_{\ubc} \}$.  We will then write $A_t$ for an allocation rule conditional on the cost being $c_t$, and $\pi_t$ the probability of the cost of an agent being $c_t$.  We will assume that $\bar{B} < \sum_{t=1}^{\ubc} \pi_t c_t$, meaning that it is not feasible to accept all data points, since otherwise it is trivially optimal to set $A_t = 1$ for all $t$. 

The following lemma describes the optimization problem faced by an analyst wanting to design an optimal survey mechanism.  Recall that residual values lie in the interval $[L,H]$.

\begin{lemma}[Optimization Problem for Parameter Estimation]\label{lem: regression-program}
The optimization program for the analyst is given by:
\begin{equation}\label{eq: regression-program}
\begin{aligned}
\inf_{A \in [0,1]^{\ubc}}~&\sup_{\duals \in [0,1]^{\ubc}} \sum_{t=1}^l \frac{\pi_t}{A_t} \left((1-\duals_t) \cdot  L^2  + \duals_t \cdot U^2 \right)\\
\mathrm{s.t.}~~~&\sum_{t=1}^{\ubc} \pi_t \left( (1-\duals_t) \cdot L+ \duals_t \cdot U  \right)= 0\\
~& \sum_{t=1}^{\ubc} \pi_t c_t A_t \leq \bar{B}\\
~& A \text{ is monotone non-increasing }\\
\end{aligned}
\end{equation}
\end{lemma}

\begin{proof}[Proof of Lemma~\ref{lem: regression-program}]
First we note that
\begin{equation}
\begin{aligned}
\cR^*(\cF, A) \triangleq  \left( \sup_{\cX} \sum_{j=1}^d \E[x_i x_i^\top]_{jj}^{- 1} \right) \cdot~&\sup_{\cD_\epsilon}~\E \left[ \epsilon_i^2 \frac{1 \{i \in S \} }{ A^2(c_i) } \right]
\\&\text{s.t.}~\E[\epsilon_i] = 0   
\end{aligned}
\end{equation}
We can therefore renormalize the worst-case variance by $\sup_{\cX} \sum_{i=1}^d \E[x_i x_i^\top]_{ii}^{- 1}$, as it does not depend on any other parameter of the problem. The analyst's objective is now given by 
\begin{equation}
\begin{aligned}
&\sup_{\cD_\epsilon}~\sum_{t=1}^{\ubc} \pi_t \frac{ \E[\epsilon_i^2|c_t]}{A_t}
\\&\text{s.t.} \E[\epsilon_i] = 0   
\end{aligned}
\end{equation}
The worst case distribution is reached when $\epsilon_i|c_t$ is binomial between $L$ and $U$ (and such a distribution is feasible for $\epsilon_i|c_t$), therefore letting $\duals_t = P[\epsilon_i = U \mid c_t]$, we obtain the lemma. 
\end{proof}
We can now characterize the form of the optimal survey mechanism.  For simplicity, we will assume that $U^2 \geq L^2$.  This is without loss of generality, since the optimization program is symmetric in $L$ and $U$; if $L^2 > U^2$, the analyst can set $q_t = 1-q_t$, $L=U$ and $U=L$ to obtain Program~\eqref{eq: regression-program} with $U^2 > L^2$. 

\begin{theorem}\label{thm:main-regression}  
Under the assumptions described above, an optimal allocation rule $A$ has the form
\begin{enumerate}
\item $A_t = \min\left(1,\alpha \frac{|L|}{\sqrt{c_t}}\right)$ for $t < t^-$
\item $A_t = \bar{A}$ for all $t \in \{ t^-, \ldots, t^+\}$
\item $A_t = \min\left(1,\alpha \frac{U}{\sqrt{c_t}}\right)$ for $t > t^+$
\end{enumerate}
for $\bar{A}$ and $\alpha$ positive constants that do not depend on $n$, and $t^-$ and $t^+$ integers with $t^- \leq t^+$. Further, $\bar{A}$ and $\alpha$ can be computed efficiently given knowledge of $t^-,t^+$.
\end{theorem}

We remark that the allocation rule that we designed is strictly positive and independent of $n$ (as the optimization program itself does not depend on $n$), so Lemmas~\ref{lem: linear-consistent-estimator} and~\ref{lem: linear-asymptotic-variance} apply. Theorem~\ref{thm:main-regression} immediately implies that an optimal allocation rule can be obtained by simply searching over the space of parameters $(t^-,t^+)$, which can be done in at most $\ubc^2$ steps. For each pairs of parameters $(t^-,t^+)$, $A$ can be computed efficiently as stated in the Theorem. Then the analyst only needs to pick the allocation rule that minimizes the objective value among the obtained allocation rules that are feasible for Program~\eqref{eq: regression-program}. Further, we remark that the solution for the linear regression case exhibits a structure that is similar to the structure of the optimal allocation rule for moment estimation (see Theorem~\ref{thm:main-moment-estimation}): it exhibits a pooling region in which all cost types are treated the same way, and changes in the inverse of the square root of the cost outside said pooling region. However, we note that we may now choose to pool agents together in an intermediate range of costs, instead of pooling together agents whose costs are below a given threshold. 

\begin{proof}[Proof sketch]
We first compute the best response $\duals^*$ of the adversary; we note that this best response is in fact the solution to a knapsack problem that is independent of the value taken by the allocation rule $A$. We can therefore plug the adversary's best response into the optimization problem, and reduce the minimax problem above in a simple minimization problem on $A$. We then characterize the solution as a function of the parameters $(t^-,t^+) \in [\ubc]^2$ through KKT conditions. The full proof is given in Appendix~\ref{app: regression}.
\end{proof}

\paragraph{Non-linear regression:} We further show in the Appendix~\ref{app:nonlinear-regression} that our results extend to non-linear regression, i.e. when $y_i$ is generated by a process of the more general form 
$$y_i = f(\theta^*,x_i) + \epsilon_i,$$
 under a few additional assumptions on the distribution of $x$ and on the regression function $f$.

\section*{Acknowledgments}

Yiling Chen was partially supported by NSF grant CCF-1718549. Juba Ziani was supported by NSF grants CNS-1331343 and CNS-1518941, and the US-Israel Binational Science Foundation grant 2012348. Part of the work was done while Yiling Chen and Juba Ziani were at Microsoft Research New England.

\bibliographystyle{abbrv}
\bibliography{bibliography}

\newpage 

\begin{appendix}

\section{Extension: Multi-dimensional non-linear regression}\label{app:nonlinear-regression}
In this section, we consider a $d$-dimensional non-linear regression setting. As in the linear regression setting, we let $d$ be the dimension of the parameter and $\theta^* \in \text{int}\left( \Theta \right)$ be the parameter to estimate, where $\Theta$ is a compact subset of $\reals^d$ and the operator int(.) denotes the interior of a set. Each agent $i \in [n]$ has a cost $c_i$ taken i.i.d from a known distribution $\cF$ on discrete support $\cC$, and draws a pair of data points $(x_i,y_i)$ in the following manner: first, the agent draws $x_i \in \reals^d$ according to an unknown distribution, such that $x_i$ is independent of the cost $c_i$; the draws are i.i.d. among agents. Then, $y_i \in \reals$ is generated according to the following process:
\begin{equation}
    y_i = f(\theta^*,x_i) + \epsilon_i
\end{equation}
with $\epsilon_i \in [L, U]$ is a bounded random variable with known mean $\E[\epsilon_i] =  0$ (hence $L \leq 0 \leq U$) and is correlated with the agent's cost $c_i$ for revealing his data, but independent of $x_i$. The $\epsilon_i$'s are also drawn i.i.d among agents. 

As before, the analyst wants to design a mechanism to buy from the agents, then compute an estimate $\hat{\theta}_S$ of $\theta$ where $S$ is the set of elicited agents, while not paying each more than the total budget $\bar{B}$ per agent in expectation. Let $S$ be the set of elicited data points. The analyst's estimate is the value $\hat{\theta}_S$ that minimizes the Horvitz-Thompson estimator of the moment, i.e.
\begin{equation}
 \hat{\theta}_S = \argmin_{\theta \in \reals^d} \sum_i \frac{1 \{i \in S\}}{A(c_i)} \left(y_i-f(\theta,x_i) \right)^2
\end{equation}
We make a few assumptions on function $f$. First, we require that it is twice continuously differentiable:
\begin{assumption}[Twice continuous differentiability of $f$]\label{as: f-continuous}
For any $x \in \reals^d$, $\theta \to f(\theta,x)$ is twice continuously differentiable.
\end{assumption}

Further, we want the function $f$ and its derivative to exhibit a non-degeneracy property; i.e., $f$ cannot have the same distribution for different values of $\theta$:
\begin{assumption}[Non-degeneracy]\label{as: non-degenerate}
For any $\theta \neq \theta^*$, $\E_x\left[  \left( f(\theta^*,x)  - f(\theta,x) \right)^2 \right] > 0$. Equivalently, $f(\theta^*,x)  \neq f(\theta,x)$ with positive probability, where the randomness is taken on the distribution of $x$. 
\end{assumption}

We require that the covariance matrix of the gradient of $f$ in $\theta$ is invertible at $\theta^*$, which is a generalization of the covariance matrix of features being invertible for the linear regression case:
\begin{assumption}\label{as: positive-definite}
$\E \left[\nabla_\theta f(\theta^*,x_i) \nabla_\theta f(\theta^*,x_i)^\top \right]$ is positive-definite, therefore invertible.
\end{assumption}

Finally, we require that the Hessian of $\left(y_i-f(\theta,x_i) \right)^2$ in $\theta$ is bounded by a function with finite expectation in a neighborhood of $\theta^*$, which generalizes the condition on the covariance matrix of features being finite in the linear regression case:
\begin{assumption}\label{as: bounded-second-derivative}
There exists a function $g$ that does not depend on $\theta$ and has finite expectation, such that
$$\Vert \nabla^2_{\theta \theta} \left(y_i-f(\theta,x_i) \right)^2 \Vert  \leq g(x_i,y_i,c_i,S)~~\forall \theta \in \mathcal{N}$$ 
where $\mathcal{N}$ is a compact neighborhood of $\theta^*$.
\end{assumption}

We proceed to show that $\hat{\theta}_S$ is a consistent estimator of $\theta^*$:
\begin{lemma}\label{lem: nonlinear_consistent_estimator}
Under Assumptions~\ref{as: f-continuous} and~\ref{as: non-degenerate}, for any allocation rule $A > 0$ that does not depend on $n$, $\hat{\theta}_S$ is a consistent estimator of $\theta^*$.
\end{lemma}
\begin{proof}
 Let $m(\theta;x,y) = (y-f(\theta,x))^2$, and $w_i = 1\{ i \in S \}$ for simplicity. The following hold:
\begin{enumerate}
\item First, we note that $\theta^*$ is the unique parameter that minimizes $\E [(y_i-f(\theta,x_i))^2]$; indeed, take any $\theta \neq \theta^*$, we have that 
\begin{align*}
\E [(y_i-f(\theta,x_i))^2]
&= \E \left[ \left(y_i-f(\theta^*,x_i) + f(\theta^*,x_i)  - f(\theta,x_i) \right)^2 \right] 
\\& = \E \left[ \left(y_i-f(\theta^*,x_i) \right)^2 \right] + \E\left[  \left( f(\theta^*,x_i)  - f(\theta,x_i) \right)^2 \right] + 2 \E\left[ \epsilon \left( f(\theta^*,x_i)  - f(\theta,x_i) \right) \right]
\\& = \E \left[ \left(y_i-f(\theta^*,x_i) \right)^2 \right] + \E\left[  \left( f(\theta^*,x_i)  - f(\theta,x_i) \right)^2 \right] 
\\& > \E \left[ \left(y_i-f(\theta^*,x_i) \right)^2 \right]
\end{align*}
where the second to last step follows from independence of $\epsilon$ from $x$ and the parameter $\theta$ and the fact that $\E[\epsilon] = 0$, and the last step follows from Assumption~\ref{as: non-degenerate}. 
\item $\Theta$ is a compact subset of $\reals$.
\item $m(\theta;x,y)$ is continuous in $\theta$ and so is its expectation.
\item $m(.;.)$ is also bounded (lower-bounded by $0$, and upper-bounded by either $L^2$ or $U^2$). Therefore, $\theta \to \frac{w_i }{A(c_i)} m(\theta;x_i,y_i)$ is continuous and bounded, thus by the uniform law of large number,  $\sup_{\theta \in \Theta} \mid\frac{1}{n} \sum_{i=1}^n \frac{w_i }{A(c_i)} m(\theta;x_i,y_i) - \E[m(\theta;x_i,y_i)]| \to 0$. The expectation term comes from remembering that conditional on $c_i$, $m(\theta;x_i,y_i)$  and $\frac{w_i }{A(c_i)}$ are independent hence 
$$\E\left[\frac{w_i}{A(c_i)} m(\theta;x_i,y_i) \right]= \E\left[ \Econd{ \frac{w_i }{A(c_i)} } {c_i} \Econd{ m(\theta;x_i,y_i)}{c_i} \right] = \E\left[ m(\theta;x_i,y_i) \right] $$
noting that $\Econd{ \frac{w_i }{A(c_i)} } {c_i} = 1$.
\end{enumerate}
Therefore, all of the conditions of Theorem 2.1 of~\cite{McFadden86} are satisfied, which is enough to prove the result.
\end{proof}

The following theorem characterizes the asymptotic variance of $\hat{\theta}_S$.  (In fact, it fully characterizes the asymptotic distribution of $\hat{\theta}_S$.)
\begin{lemma}\label{lem: asymptotic-variance}
Under Assumptions~\ref{as: f-continuous},~\ref{as: non-degenerate},~\ref{as: positive-definite} and~\ref{as: bounded-second-derivative}, for any allocation rule $A > 0$ that does not depend on $n$, the asymptotic distribution of $\hat{\theta}_S$ is given by 
\begin{align*}
\sqrt{n} (\hat{\theta}_S - \theta^*) \xrightarrow[]{d} \mathcal{N}\left(0,  \E \left[\nabla_\theta f(\theta^*,x_i) \nabla_\theta f(\theta^*,x_i)^\top \right]^{-1} \E \left[\frac{w_i }{ A^2(c_i) } \epsilon_i^2 \right] \right)
\end{align*}
where $d$ denotes convergence in distribution and where randomness in the expectations is taken on the costs $c_i$, the set of elicited data points $S$, the features of the data $x_i$, and the noise $\epsilon_i$.
\end{lemma}

\begin{proof}
For simplicity, let $w_i = 1 \{ i \in S \}$ and note that the $w_i$'s are i.i.d. Let $m(\theta;x_i,y_i) \cdot \frac{w_i}{A(c_i)} = \frac{w_i}{A(c_i)}  (y_i-f(\theta,x_i))^2$. First we remark that 
$$\nabla_\theta m(\theta;x_i,y_i) \cdot \frac{w_i}{A(c_i)} = 2 \frac{w_i}{A(c_i)} (f(\theta,x_i)-y_i) \nabla_\theta f(\theta,x_i)$$  
and 
$$\nabla^2_{\theta \theta} m(\theta;x_i,y_i) \cdot \frac{w_i}{A(c_i)} = 2\frac{w_i}{A(c_i)} \left( \nabla_\theta f(\theta,x) \nabla_\theta f(\theta,x_i)^\top +  (f(\theta,x_i)-y_i)\cdot \nabla^2_{\theta \theta} f(\theta,x_i)  \right).$$ 
We then note the following:
\begin{enumerate}
\item $\theta^*$ is in the interior of $\Theta$.
\item $\theta \to m(\theta;x_i,y_i) \cdot \frac{w_i}{A(c_i)} $ is twice continuously differentiable for all $x_i,y_i,c_i,w_i$.
\item $\sqrt{n} \sum_{i=1}^n \nabla_\theta m(\theta^*;x_i,y_i) \cdot \frac{w_i}{A(c_i)} \to \mathcal{N}\left(0,4 \E \left[\frac{w_i}{A^2(c_i)} (f(\theta^*,x_i) - y_i)^2 \nabla_\theta f(\theta^*,x_i) \nabla_\theta f(\theta^*,x_i)^\top \right] \right) $. This follows directly from applying the multivariate central limit theorem, noting that 
$$
\E \left[ \nabla_\theta m(\theta^*;x_i,y_i) \cdot \frac{w_i}{A(c_i)} \right] = \E \left[\Econd{2 \epsilon_i \nabla_\theta f(\theta^*,x_i)}{ c_i} \cdot \Econd{\frac{w_i}{A(c_i)} }{ c_i} \right]  =  \E\left[2 \epsilon_i \nabla_\theta f(\theta^*,x) \right] = 0 
$$
where the first step follows from conditional independence on $c$ of $x,\epsilon$ with $A(c), S$, the second step from $ \E \left[ \frac{w_i}{A(c_i)} \mid c_i \right] = 1$, and the last equality follows from the fact that $x, \theta$ are independent of $\epsilon$ and $\E[\epsilon_i ] = 0$.
\item $\sup_{\theta \in \mathcal{N}} \left \Vert \E \left[ \nabla_{\theta \theta}^2 m(\theta;x_i,y_i) \cdot \frac{w_i}{A(c_i)} \right] - \frac{1}{n} \sum_i \nabla_{\theta \theta}^2 m(\theta;x_i,y_i) \cdot \frac{w_i}{A(c_i)} \right \Vert \to 0$, applying the uniform law of large numbers as $\nabla_{\theta \theta}^2 m(\theta;x_i,y_i) \cdot \frac{w_i}{A(c_i)}$  is continuous in $\theta$ and as it is dominated by a function that is independent of $\theta$ and has finite expectation, by Assumption~\ref{as: bounded-second-derivative}.
\item $\E \left[ \nabla_{\theta \theta}^2 m(\theta;x_i,y_i) \cdot \frac{w_i}{A(c_i)} \right]$ is positive-definite and invertible at $\theta^*$; indeed, it is given by
\begin{align*}
&2\E \left[ \frac{w_i}{A(c_i)} \nabla_\theta f(\theta^*,x_i)  \nabla_\theta f(\theta^*,x_i)^\top \right] + 2\E \left[  \frac{w_i}{A(c_i)} (f(\theta^*,x_i) - y_i) \cdot \nabla^2_{\theta \theta} f(\theta^*,x_i)  \right]
\\&= 2 \E \left[\nabla_\theta f(\theta^*,x_i) \nabla_\theta f(\theta^*,x_i)^\top\right]  -2 \E \left[  \Econd{\frac{w_i}{A(c_i)}}{c_i} \Econd{ \epsilon_i\cdot \nabla^2_{\theta \theta} f(\theta^*,x_i)}{c_i}  \right]
\\&=  2\E \left[\nabla_\theta f(\theta^*,x_i) \nabla_\theta f(\theta^*,x_i)^\top\right] - 2\E \left[\epsilon_i \nabla^2_{\theta \theta} f(\theta^*,x_i)\right]
\\&=  2\E \left[\nabla_\theta f(\theta^*,x_i) \nabla_\theta f(\theta^*,x_i)^\top\right] - 2\E \left[\epsilon_i\right] \E \left[ \nabla^2_{\theta \theta} f(\theta^*,x_i)\right]
\\& =  2\E \left[\nabla_\theta f(\theta^*,x_i) \nabla_\theta f(\theta^*,x_i)^\top \right] 
\end{align*}
which is positive-definite by assumption.
\end{enumerate}
Therefore the sufficient conditions i)-v) in Theorem 3.1 of~\cite{McFadden86} hold, proving that the asymptotic distribution is normal with mean $0$ and variance 
$$
 \E \left[\nabla_\theta f(\theta^*,x_i) \nabla_\theta f(\theta^*,x_i)^\top \right]^{- 1}  \E \left[\frac{w_i }{ A^2(c_i) } \epsilon_i^2 \nabla_\theta f(\theta^*,x_i) \nabla_\theta f(\theta^*,x_i)^\top \right]  \E \left[\nabla_\theta f(\theta^*,x_i) \nabla_\theta f(\theta^*,x_i)^\top \right]^{-1}
$$
We conclude the proof by noting that $x$ is independent of all other random variables of the problem, hence
\begin{align*}
\E \left[\frac{w_i }{ A^2(c_i) } (f(\theta^*,x_i) - y)^2 \nabla_\theta f(\theta^*,x_i) \nabla_\theta f(\theta^*,x_i)^\top \right] 
&= \E \left[\frac{w_i }{ A^2(c_i) } \epsilon_i^2 \right] \E \left[\nabla_\theta f(\theta^*,x_i) \nabla_\theta f(\theta^*,x_i)^\top \right] 
\end{align*}
\end{proof}

We therefore have that the optimization problem solved by the analyst, whose goal is to minimize the $\sqrt{n}$-normalized worst-case asymptotic mean-squared error of $\hat{\theta}_S$ as the population size goes to infinity, can be written as:
\begin{lemma}[Optimization Problem for Parameter Estimation]
The optimization program for the analyst is given by:
\begin{equation}
\begin{aligned}
\inf_{A \in [0,1]^{\ubc}}~&\sup_{\duals \in [0,1]^{\ubc}} \sum_{t=1}^l \frac{\pi_t}{A_t} \left((1-\duals_t) \cdot  L^2  + \duals_t \cdot U^2 \right)\\
\mathrm{s.t.}~~~&\sum_{t=1}^{\ubc} \pi_t \left( (1-\duals_t) \cdot L+ \duals_t \cdot U  \right)= 0\\
~& \sum_{t=1}^{\ubc} \pi_t c_t A_t \leq \bar{B}\\
~& A \text{ is monotone non-increasing }\\
\end{aligned}
\end{equation}
\end{lemma}

\begin{proof}
The proof is identical to that of Theorem~\ref{lem: regression-program}, and therefore is omitted.
\end{proof}

 We note that this is exactly the same optimization program as for the linear regression case, and therefore refer to Section~\ref{sec:linear-regression} for results on how to compute the optimal allocation rule.

\section{Omitted proofs}
\subsection{Proof of the reduction from Mechanism Design to Optimization}\label{app:truth-char}
\subsubsection{Continuous support, atomless distribution}

We first give the proof of the result when the cost support is continuous and the distribution atom-less:
\begin{proof}
By truthfulness: $u(c; c) = \max_{\hat{c}} u(\hat{c}; c)$. By the envelope theorem:
\begin{equation}
\frac{\partial u(c; c)}{\partial c} = \frac{\partial u(\hat{c}; c)}{\partial c} \bigg|_{\hat{c}=c} = -A(c)
\end{equation}
Let $C$ be the upper bound on the support of the distribution of costs. The analyst will always pay such a player his true cost, hence $u(C;C)=0$. Integrating the above equation from $c$ to $C$ we get:
\begin{equation}\label{eqn:payment-char}
U(C;C) - u(c;c) = - \int_{c}^C A(z)\cdot dz  \Rightarrow u(c;c) = \int_{c}^C A(z) dz \Rightarrow P(c) = c + \frac{1}{A(c)} \int_{c}^C A(z) dz
\end{equation}
Thus the allocation rule $A$ uniquely determines the truthful payment rule that accompanies it. Moreover, if the allocation rule is monotone, then it can be shown that indeed the pair of allocation $A$ and payment $P$ defined by Equation \eqref{eqn:payment-char} define a TIR mechanism (identical to the analysis in \cite{Myer81}). The expected payment of a mechanism can then be computed as:
\begin{align*}
\E\left[ P(c) \cdot A(c) \right] =~& \E\left[ c\cdot A(c) \right] + \int_{0}^C \int_{c}^C A(z) dz f(c) dc = \E\left[ c\cdot A(c) \right] +\int_{0}^C A(z) \int_{0}^c f(c) dc dz\\
 =~& \E\left[ c\cdot A(c) \right] + \int_{0}^C A(z) F(z) dz =  \E\left[ c\cdot A(c) \right] + \E\left[ A(c) \frac{F(c)}{f(c)}\right] = \E\left[\phi(c)\cdot A(c)\right]
\end{align*}
Thus the budget constraint of the mechanism design problem is equal to the budget constraint of the optimization problem by replacing $c$ with a virtual cost $\phi(c)$ and with allocation function $A'$ on the virtual cost space defined by $A'(\phi(c))=A(c) \Leftrightarrow A'(c) = A(\phi^{-1}(c))$. 

Finally, since the estimator and the variance of the estimator do not depend on the realized costs, other than through the allocation function associated with them, we get that the variance of the estimator under $A$ in the cost space, will be the same as the variance of the estimator under $A'=A\circ \phi^{-1}$ in the virtual cost space, i.e. $\Var^*(\hat{\theta}_S; \phi(\cF), A\circ \phi^{-1}) = \Var^*(\hat{\theta}_S; \cF, A)$. Thus if we solve the optimization problem for distribution of costs $\phi(\cF)$ we will get back a solution $\tilde{A}^*$, with variance $\Var^*(\hat{\theta}_S; \phi(\cF), \tilde{A}^*)\leq \Var^*(\hat{\theta}_S; \phi(\cF), A\circ \phi^{-1})$. Then the allocation $A^*=\tilde{A}^*\circ \phi$ can be coupled with the payment rule in Equation \eqref{eqn:payment-char} and get a TIR mechanism, which also respects the budget constraint, since:  $\E_{c \sim \cF}[P(c)\cdot \tilde{A}^*(\phi(c))] = \E[\phi(c)\cdot\tilde{A}^*(\phi(c))] =\E_{x \sim \phi(\cF)}[x\cdot \tilde{A}^*(x)]=\bar{B}$. The latter equality holds by the fact that $A^*$ was a feasible allocation rule for optimization problem under distribution of costs $\phi(\cF)$.
\end{proof}

\subsubsection{Discrete cost support}

We now give the proof when the cost support is discrete. In the whole proof, we let $\cC \triangleq \{ c_1, \ldots, c_{\ubc}\}$ with $c_1 < \ldots < c_{\ubc}$, $\pi_t = f(c_t)$, and $A_t = A(c_t)$. We first show that given a fixed monotone non-increasing allocation rule $A$ with $A_1 \geq \ldots \geq A_{\ubc}$, there exists optimal prices $P^*(A)$ such that the payments of any individually rational and truthful mechanism are lower-bounded by $P^*(A)$, and such that mechanism with allocation rule $A$ and prices $P^*(A)$ is individually rational and truthful: 

\begin{claim}
Let $P^*_{\ubc}(A) = c_{\ubc}$, and $P^*_{t}(A) = c_t + \sum_{j=t+1}^{\ubc} \frac{A_{j}}{A_{t}} (c_j - c_{j-1})$ for all $t <\ubc$. Then for every IC and IR mechanism with monotone non-increasing allocation rule $A$, the pricing rule $P$ must satisfy $P_{t} \geq P^*_{t}(A)$ for all $t$. Further, the mechanism with allocation rule $A$ and pricing rule $P^*(A)$ is IC and IR. 
\end{claim}

Note that this directly implies that if there exists a variance-minimizing mechanism, then there exists an IC and IR variance-minimizing mechanism with pricing rule $P^*(A)$ given allocation rule $A$. Therefore, we can reduce our attention to such mechanisms.

\begin{proof}
We show the first part of the lemma by induction: clearly, it must be the case that $p_{\ubc} \geq c_{\ubc}$ for the mechanism to be IR. Now, suppose by induction that for any IC and IR mechanism, $p_{t+1} \geq P^*_{t+1}(A)$. We require by IC constraint and induction hypothesis that for $t < \ubc-1$,
\begin{align*}
p_t \geq  c_t + \frac{A_{t+1}}{A_t} (P^*_{t+1}(A) - c_t) 
=  c_t + \frac{A_{t+1}}{A_t} \left(\sum\limits_{j=t+2}^{\ubc} \frac{A_{j}}{A_{t+1}} (c_j - c_{j-1}) + c_{t+1}- c_t \right)
= P^*_{t}(A) 
\end{align*}
and for $t= \ubc-1$,
\begin{align*}
p_{\ubc-1} \geq  c_{\ubc-1} + \frac{A_{\ubc}}{A_{\ubc-1}} (P^*_{\ubc}(A) - c_{\ubc-1}) 
=  c_{\ubc-1} + \frac{A_{\ubc}}{A_{\ubc-1}} (c_{\ubc} - c_{\ubc-1}) 
= P^*_{\ubc-1}(A) 
\end{align*}

This proves the first part of the claim. Now, we note that the mechanism with prices $P^*(A)$ is IR as clearly $P^*_{t}(A) \geq c_t$. It remains to show the mechanism is IC to complete the proof of the claim. Take $t \neq t'$, we have:
\begin{align*}
A_t (P^*_t - c_t) - A_{t'} (P^*_{t'} - c_t) 
& = \sum\limits_{j=t+1}^{\ubc} A_{j} (c_j - c_{j-1}) - \sum\limits_{j=t'+1}^{\ubc} A_{j} (c_j - c_{j-1}) + A_{t'} (c_t - c_{t'})
\end{align*}
If $t > t'$, we have 
$$A_t (P^*_t - c_t) - A_{t'} (P^*_{t'} - c_{t'}) \geq A_{t'} (c_t - c_{t'}) - A_{t'} \sum\limits_{j=t'+1}^{t} (c_j - c_{j-1}) = 0$$ 
as $A_j \leq A_{t'}$ for all $j \geq t'$, while if $t < t'$, we have 
$$A_t (P^*_t - c_t) - A_{t'} (P^*_{t'} - c_t) \geq  A_{t'} \sum\limits_{j=t+1}^{t'} (c_j - c_{j-1}) - A_{t'} (c_{t'} - c_t)  = 0$$ 
as $A_j \geq A_{t'}$ for all $j \leq t'$. This concludes the proof.
\end{proof}

We conclude the proof by showing that the budget constraint can be rewritten in the desired form, i.e. such that the true costs are replaced by the virtual costs in the budget expression: 
\begin{lemma}
The expected budget used by a mechanism with allocation rule $A$ and payment rule $P^*(A)$ can be written
\begin{align*}
\sum_{t=1}^n \sum_{t=1}^{\ubc} \pi_t \phi(c_t) A_t
\end{align*}
\end{lemma}

\begin{proof}
The expected budget spent on an agent can be written, using the previous claim:
\begin{align*}
\sum_{t = 1}^{\ubc} \pi_t  P^*_t(A) A_t
 &= \sum_{t = 1}^{\ubc} \pi_t c_t A_t+ \sum_{t=1}^{\ubc} \pi_t \sum_{j=t+1}^{\ubc} A_{j} (c_j - c_{j-1})
\\&=  \sum_{t = 1}^{\ubc} \pi_t c_t A_t+ \sum_{j=2}^{\ubc} \sum_{t=1}^{j-1} \pi_t A_{j} (c_j - c_{j-1}) 
\\&= \pi_1 c_1 A_1  + \sum_{j=2}^{\ubc} \left(\pi_j c_j +  (c_j - c_{j-1}) \sum_{t=1}^{j-1} \pi_t\right)  A_j
\\&= \sum_{j=1}^{\ubc} \pi_j  \phi(c_j)  A_j 
\end{align*}
\end{proof}

\subsection{Proof of Theorem \ref{thm:main-moment-estimation}}\label{app:proof-of-moment-estimation}

\subsubsection{A More Structural Characterization}
We begin by giving a strengthening of our main Theorem \ref{thm:main-moment-estimation} that exactly pin-points the optimal allocation rule in a closed form. We then give a proof of this stronger theorem.

\begin{theorem}[Closed Form for Optimal Allocation for Moment Estimation]\label{thm:main-moment-estimation-closed}
The optimal allocation rule $A$ is determined by two constants $\bar{A}$ and $t^*\in \{0,\ldots, \ubc\}$ such that:
\begin{equation}\label{eqn:generic-form-appendix}
A_t = \begin{cases}
\bar{A} & \text{if $t\leq t^*$}\\
\frac{1}{\sqrt{c_t}}\cdot \frac{\bar{B} - \bar{A} \E[c\cdot 1\{c\leq c_{t^*}\}]}{\E[\sqrt{c}\cdot 1\{c>c_{t^*}\}} & \text{o.w.}
\end{cases}
\end{equation}
The parameters $\bar{A}$ and $t^*$ are determined as follows. For $k\in \{0, 1,\ldots,\ubc+1\}$ and $\interp\in [0,1]$ let (for $c_0=0$ and $c_{\ubc+1}=\infty$):
\begin{align*}
Q(k, \interp) =~& \sum_{t=1}^k \pi_t c_t + \sum_{t=k+1}^{\ubc} \pi_t \sqrt{\frac{c_t\cdot c_k}{\interp}} \Label{eqn:q-def} &
R(k, \interp) =~& 2\left(\sum_{t=1}^{k} \pi_t \frac{c_t\cdot \interp}{c_{k}} + \sum_{t=k+1}^{\ubc} \pi_t \right) \Label{eqn:r-def}
\end{align*}
\begin{align*}
~~~~~~~~~~~~~~~~~~~~~~~~~~~~~~~~~~~~~~~~B(k, \interp) =~& \frac{Q(k, \interp)}{R(k, \interp)} \Label{eqn:b-def}~~~~~~~~~~~~~~~~~~~~~~~~~~~~~~~~~~~~~~~~
\end{align*}
Let $k^*$ be the unique $k$ s.t. $B(k,1)\leq \bar{B} < B(k+1,1)$. If $k^*=0$ then $t^*=0$. Otherwise let $\interp^*\in [0,1]$ be the unique solution to: $
\bar{B} = B(k^*, \interp^*)$.\footnote{The latter amounts to solving a simple cubic equation of the form $\frac{A}{\sqrt{x}} + B = x C + D \Leftrightarrow \sqrt{x}^3 C + (D-B)\sqrt{x} - A=0$, which admits a closed form solution.}
If $R(k^*, \interp^*)\geq 1$ then $t^*=k^*$ and $\bar{A}=\frac{1}{R(k^*,\interp^*)}$. If $R(k^*, \interp^*)<1$ then $t^*=\max\{k:\bar{B} > Q(k,1)\}$ and $\bar{A}=1$.
\end{theorem}

\subsubsection{Proof of Theorem \ref{thm:main-moment-estimation-closed}: Optimal Survey for Moment Estimation}

In all that follows, we will drop the monotonicity constraints on $A$. We write 
\begin{equation}
\begin{aligned}
P = \inf_{A\in \left(0,1\right]^{\ubc}}\sup_{\duals \in [0,1]^{\ubc}}~&
\sum_{t=1}^{\ubc} \pi_t \cdot \frac{\duals_t}{A_t} - \left(\sum_{t=1}^{\ubc} \pi_t\cdot \duals_t\right)^2\\
\mathrm{s.t.}~& \sum_{t=1}^{\ubc} \pi_t \cdot c_t \cdot A_t  \leq \bar{B}\\
\end{aligned}
\end{equation}
We will show that nevertheless, the solution to this optimization program satisfies the monotonicity constraint in $A$, hence dropping the constraint can be done without loss of generality. For further simplification of notation, for any two vectors $x, y$ we let $x\cmult y$ be their component-wise product vector, $x\cdiv y$ their component-wise division  and $\ldot{x}{y}$ their inner product. Finally  we denote with $x_{i:k}$ the sub-vector $(x_i, x_{i+1},\ldots, x_k)$. Thus we can write the objective function inside the minimax problem as:
\begin{equation}
V(A, \duals)=\ldot{\pi}{\duals\cdiv A} - \ldot{\pi}{\duals}^2
\end{equation}
where $\pi$ is the pdf vector and we can write the budget constraint as $\langle \pi \cmult c, A\rangle\leq \bar{B}$.

Rather than solving for simply the optimal solution for $A$ in the latter minimax problem, we will instead address the harder problem of finding an equilibrium of the zero-sum game $\cG$ associated with this minimax, i.e. a game where the minimizer player is choosing $A$ and the maximizing player is choosing $\duals$ and the utility is $V(A,\duals)$. An equilibrium of this game is then defined as:
\begin{defn}[Equilibrium Pair] A pair of solutions $(A^*, \duals^*)$ is an equilibrium if:
\begin{align}
V(A^*, \duals^*) =~& \inf_{A \in [0,1]^{\ubc}: \ldot{\pi \cmult c}{A}  \leq \bar{B}} V(A, \duals^*)=\sup_{\duals \in [0,1]^{\ubc}} V(A^*, \duals)
\end{align}
\end{defn}
Observe that the function $V(A,\duals)$ is convex in $A$ and concave in $\duals$, hence it defines a convex-concave zero-sum game. From standard results on zero-sum games, if $(A^*, \duals^*)$ is an equilibrium solution, then $A^*$ is a solution to the minimax problem $P$ that we are interested in (see e.g. \cite{Freund1999}), since:
\begin{align*}
\inf_{A} V(A,\duals^*) \leq \inf_{A } \sup_{\duals} V(A,\duals) \leq \sup_{\duals} V(A^*,\duals) = V(A^*,\duals^*) = \inf_{A} V(A,\duals^*)
\end{align*}
directly implying $\sup_{z} V(A^*, z) = \inf_{A } \sup_{\duals} V(A, \duals)$.

\subsubsection{Characterizing the best responses of the minimizing and maximizing player}

In this paragraph we characterize the best-responses of the minimizing and maximizing players in the zero-sum game formulation of our problem.

\begin{lemma}[Best Response of Min Player]\label{lem: best-resp-A}
Fix $\duals$ such that $\duals_t > 0$ for all $t$. Let $\lambda^* > 0$ be such that
\begin{align}\label{eqn:bs-binding-budget}
\sum_{t = 1}^{\ubc} \pi_t\cdot c_t \cdot \min \left(1,\sqrt{\frac{\duals_t}{\lambda^* c_t}} \right) = \bar{B},
\end{align}
Then the allocation rule $A$ such that $A_t = \min \left(1,\sqrt{\frac{\duals_t}{\lambda^* c_t}} \right)$ is a best-response to $\duals$ in game $\cG$.
\end{lemma}

\begin{proof}
Fix $\duals$ . Note that when $A_t$ goes to $0$ for any $t$, the objective values tends to infinity, but the optimal solution is clearly finite (simply splitting the budget evenly among cost types $c_t$ is feasible and leads to a finite objective value). This implies that a best-response exists for the analyst: indeed, the objective is convex and continuous on $(0,1]^{\ubc}$, and the feasible set can be restricted without loss of generality to be compact and convex by adding the constraints $A_t \geq \gamma$ for a small enough $\gamma > 0$. The Lagrangian (for more on Lagrangians and KKT conditions, see~\cite{Boyd04}) of the minimization problem solved by the minimizing player is therefore given by:
\begin{align*}
\cL(A,\lambda, \lambda_t^1) =  V(A, \duals)+ \lambda \sum_{t} (\pi_t c_t A_t -\bar{B}) + \sum_{t=1}^{\ubc} \lambda_t^1 (A_t -1) 
\end{align*}
where $\lambda,\lambda_t^1 \geq 0$. Let 
$(\lambda^*,\lambda_t^{1*})$ denote optimal dual variables and $A$ an optimal primal variable, we have that $\frac{\partial \cL(A,\lambda^*,\lambda_t^{1*})}{\partial A_t} = 0$ for all $t$ by the KKT conditions, implying
\begin{align*}
- \pi_t \frac{\duals_t}{A_t^2} + \lambda^* \pi_t c_t + \lambda_t^{1*}  = 0 \Rightarrow 
A_t= \sqrt{\frac{\pi_t \duals_t}{\lambda^* \pi_t c_t + \lambda_t^{1*} }}
\end{align*}
where we note that the denominator is non-zero as we have $c_t > 0$ and $\lambda_t^{1*} \geq 0$. Further, if $A < 1$, we must  also have by the KKT conditions that $\lambda_t^{1*} = 0$. This directly implies that at a best response, we must have
$A_t = \min\left(1,\sqrt{\frac{\duals_t}{\lambda^* c_t}}\right)$, 
for some $\lambda^* \geq 0$. Since the budget constraint will always be binding (as increasing the allocation probability can only help the variance), $\lambda^*$ must be solving Equation \eqref{eqn:bs-binding-budget}. The latter concludes the proof of the Lemma.
\end{proof}

This gives the best response of the minimizing player. The best response of the maximizing player can be obtained by similar techniques: 
\begin{lemma}[Best-Response of Max Player]\label{lem: best-resp-zeta}
Fix $A$ such that $A_j > 0$ for all $j$. Take any $\duals \in [0,1]^{\ubc}$ such that for every $j$, at least one of the following holds:
\begin{enumerate}
\item $\duals_j = 0$ and $\frac{1}{A_j} < 2 \ldot{\pi}{\duals}$
\item $\duals_j = 1$ and $\frac{1}{A_j} > 2  \ldot{\pi}{\duals}$
\item $0 \leq \duals_j \leq 1$ and $\frac{1}{A_j} = 2 \ldot{\pi}{\duals} $
\end{enumerate}
Then $\duals$ is a best response to $A$ in game $\cG$.
\end{lemma}
\begin{proof}
Fix $A$. The Lagrangian of the maximization problem solved by the maximizing player is given by: 
\begin{align*}
\cL(\duals,\lambda_t^0, \lambda_t^1) = V(A,\duals)+ \sum_{t=1}^{\ubc} \lambda_t^1 (1-\duals_t) + \sum_{t=1}^{\ubc} \lambda_t^0  \duals_t
\end{align*}
For optimal primal and dual variables, the KKT conditions are given by: for all $t\in \{1,\ldots, \ubc\}$
\begin{itemize}
\item  \emph{First order.} $\frac{\partial \cL(\duals,\lambda_t^{0}, \lambda_t^{1}) }{\partial \duals_t}= 0$, which can be rewritten as: 
$
\frac{\pi_t}{A_t} - 2 \pi_t \ldot{\pi}{\duals} - \lambda_t^{1} + \lambda_t^{0}  = 0 
$.
\item \emph{Feasibility.} The variables $\duals_t, \lambda_t^1, \lambda_t^0$ are feasible, i.e. $\duals_t\in [0,1]$ and $\lambda_t^{0}, \lambda_t^{1} \geq 0$
\item \emph{Complementarity.} Either $\lambda_t^{0}= 0$ or $\duals_t = 0$.  Either $\lambda_t^{1} = 0$ or $\duals_t = 1$
\end{itemize}
Because the objective value is convex and differentiable in $\duals$, the KKT conditions are necessary and sufficient for optimality of the primal and dual variables (see~\cite{Boyd04}). It is therefore enough to show that $(\duals_t,\lambda_t^0,\lambda_t^1)$ satisfies the KKT conditions for some well-chosen $\duals_t, \lambda_t^0,\lambda_t^1$. If $\frac{1}{A_t} < 2\ldot{\pi}{\duals}$, then $\duals_t=0$, $\lambda_t^1 = 0$ and $\lambda_t^0 = - \frac{\pi_t}{A_t} + 2 \pi_t \ldot{\pi}{\duals} \geq 0$ is a solution. If $\frac{1}{A_t}>2\ldot{\pi}{\duals}$ then $\duals_t=1$, $\lambda_j^0 = 0$ and $\lambda_j^1 = \frac{\pi_j}{A_j} - 2 \pi_j \ldot{\pi}{\duals} \geq 0$ is a solution. Otherwise if $\frac{1}{A_t}=2\ldot{\pi}{\duals}$ then $\lambda_t^0 = \lambda_t^1 = 0$ and any feasible value for $\duals_t$ that respects the equality $\frac{1}{A_t}=2\ldot{\pi}{\duals}$ is a solution.
\end{proof}

\subsubsection{Properties of functions $Q(k,\interp)$, $R(k,\interp)$ and $B(k,\interp)$}

Before proceeding to the main proof of the Theorem, we show some useful properties of $Q(k,\interp)$, $R(k,\interp)$ and $B(k,\interp)$. 
\begin{claim}
For all $k < \ubc$, $R(k,1) \geq R(k+1,1)$, $Q(k,1) \leq Q(k+1,1)$ and $B(k,1) \leq B(k+1,1)$. Therefore, the exists a unique $k$ such that $B(k,1) \leq \bar{B} < B(k+1,1)$.
\end{claim}

\begin{proof} By expanding:
\begin{align*}
R(k,1)- R(k+1,1) &= \sum_{t=1}^k \pi_t \frac{c_t}{c_k} + \sum_{i=k+1}^{\ubc} \pi_t  - \sum_{t=1}^{k+1} \pi_t \frac{c_t}{c_{k+1}} - \sum_{t=k+2}^{\ubc} \pi_t = \left(\frac{1}{c_k} - \frac{1}{c_{k+1}}\right) \sum_{t=1}^k \pi_t c_t  \geq 0
\end{align*}
Additionally, 
\begin{align*}
Q(k+1,1) - Q(k,1) & = \sum_{t=1}^{k+1} \pi_t c_t - \sum_{t=1}^k \pi_t c_t + \sum_{t=k+2}^{\ubc} \pi_t \sqrt{c_{k+1} c_t} - \sum_{t=k+1}^{\ubc} \pi_t \sqrt{c_k c_t} 
\\& = (\sqrt{c_{k+1}}-\sqrt{c_k}) \sum_{t=k+1}^{\ubc} \pi_t \sqrt{c_t} \geq 0
\end{align*}
This concludes the proof.
\end{proof}

\begin{claim}\label{clm: property_zeta}
There is a unique $\interp^*\in \left[\frac{c_{k^*}}{c_{k^*+1}},1\right]$ such that $\bar{B} = B(k^*,\interp^*)$ for $B(1,1) \leq \bar{B} < B(\ubc,1)$. 
\end{claim}

\begin{proof}
Observe that $B(k^*,\interp)$ is continuous decreasing in $\interp$, proving uniqueness of a solution if it exists. Existence within $\left[\frac{c_{k^*}}{c_{k^*+1}},1\right]$ follows by noting that $B(k^*,\interp)$ is continuous decreasing in $\interp$, and that trivially $B(k^*,1) \leq \bar{B} <  B(k^*+1,1) = B\left(k^*,\frac{c_{k^*}}{c_{k^*+1}}\right)$, proving the result for $B(1,1) \leq \bar{B} < B(\ubc,1)$. 
\end{proof}

\subsubsection{Proof of the Theorem} 
We split the proof of the theorem in the two corresponding cases when $B \geq B(1,1)$ i.e. $t^* \geq 1$, then deal with the corner case when $B < B(1,1)$.

\begin{lemma}[Case 1: $R(k^*,\interp^*)\geq 1$, $B \geq B(1,1)$] 
For this case, $t^*=k^*$ and $\bar{A}=\frac{1}{R(k^*,\interp^*)}= \frac{\bar{B}}{Q(k^*,\interp^*)}$.
\end{lemma}
\begin{proof}
Let $\duals$ be such that $\duals_t = \frac{c_t}{c_{k^*}} \interp^*$ for $t \leq k^*$, $\duals_{k^*+1} = \ldots = \duals_{\ubc} = 1 $. We show that $A$ defined by the parameters $\bar{A}$ and $t^*$ given in the lemma, is a best response to $\duals$ and vice-versa. 

\emph{First, let us show that $\duals$ is a best response to $A$.} For $j \leq k^*$, we note that $0 \leq \duals_j \leq 1$ and we have
$$
2 \ldot{\pi}{\duals} = 2 \left( \frac{\interp^*}{c_{k^*}} \sum_{t=1}^{k^*} \pi_t c_t + \sum_{t=k^*+1}^{\ubc} \pi_t \right) = R(k^*,\interp^*)= \frac{1}{A_j}
$$
For $t \geq k^*+1$ (when such a case exists, i.e. when $k^* < \ubc$), $\duals_j = 1$. Moreover, the allocation takes the form:
\begin{equation*}
A_t = \frac{1}{\sqrt{c_t}}\cdot \frac{\bar{B} - \bar{A} \E[c\cdot 1\{c\leq c_{t^*}\}]}{\E[\sqrt{c}\cdot 1\{c>c_{t^*}\}]} = \frac{\bar{B}}{\sqrt{c_t}}\cdot \frac{1 - \frac{1}{Q(k^*, \interp^*)} \E[c\cdot 1\{c\leq c_{t^*}\}]}{\E[\sqrt{c}\cdot 1\{c>c_{t^*}\}]}
 = \frac{\bar{B}}{\sqrt{c_t}}\cdot \frac{\sqrt{\frac{c_{k^*}}{\interp^*}}}{Q(k^*, \interp^*)}
\end{equation*}
By Claim~\ref{clm: property_zeta}, $\frac{c_{k^*+1}}{c_{k^*}} \cdot \interp^* \geq 1$. Hence:   \begin{equation}
A_t \leq \frac{\bar{B}}{\sqrt{c_t}}\cdot \frac{\sqrt{c_{k^*+1}}}{Q(k^*, \interp^*)} \leq \frac{\bar{B}}{Q(k^*, \interp^*)} = \frac{1}{R(k^*, \interp^*)} = \frac{1}{2\ldot{\pi}{q}}
\end{equation}
By Lemma~\ref{lem: best-resp-zeta}, $\duals$ is a best response to $A$. Note that  combined with the costs being non-decreasing, this also proves that $A$ is monotone non-increasing. 

\emph{It remains to show that $A$ is a best response to $\duals$.} By Lemma~\ref{lem: best-resp-A}, we only need to check that for all $j\in \{1,\ldots, \ubc\}$, $A_j = \min\left(1,\sqrt{\frac{\duals_j}{\lambda^* c_j}}\right)$ where $\lambda^*$ is chosen to make the budget constraint tight. First, we note that the $A$ given by the lemma does make the budget constraint tight, as any allocation of the form given by Equation \eqref{eqn:generic-form-appendix} does so by construction.
Now, let $\lambda^* =  \left( \frac{\sum_t \pi_t \sqrt{c_t \duals_t}}{\bar{B}} \right)^2$. We have that for $j \leq k^*$:
$$
\sqrt{\frac{\duals_j}{\lambda^* c_j}} = \sqrt{\frac{\interp^*}{c_{k^*}}} \frac{\bar{B}}{\sum_t \pi_t \sqrt{c_t \duals_t}} = \sqrt{\frac{\interp^*}{c_{k^*}}}  \frac{\bar{B}}{ \sqrt{\frac{\interp^*}{c_{k^*}}} \sum_{t=1}^{k^*} \pi_t c_t + \sum_{t=k^*+1}^{\ubc} \pi_t \sqrt{c_t}} = \frac{\bar{B}}{Q(k^*,\interp^*)} = A_j
$$
and by a similar calculation that for $j  \geq k^* + 1$ that 
$$
\sqrt{\frac{\duals_j}{\lambda^* c_j}} = \sqrt{\frac{1}{c_j}} \frac{\bar{B}}{\sum_t \pi_t \sqrt{c_t \duals_t}} =  \frac{1}{\sqrt{c_j}}\frac{\bar{B}}{\sqrt{\frac{\interp^*}{c_{k^*}}} Q(k^*,\interp^*)} = A_j
$$
Since $A_j\leq 1$ for all $j$, by the conditions of the Lemma, we therefore have: $A_j = \min\left(1,\sqrt{\frac{\duals_j}{\lambda^* c_j}}\right)$.
\end{proof}
 
\begin{lemma}[Case 2: $R(k^*,\interp^*)< 1$, $\bar{B} \geq B(1,1)$]\label{lem: case2}
For this case, $t^*=\max\{k:\bar{B} > Q(k,1)\}$ and $\bar{A}=1$.
\end{lemma}

We start with the following claim:
\begin{claim}\label{clm:tilde-B}
If $R(k^*,\interp^*)< 1$,
then there exist $\tilde{B}$, $\tilde{k}$ and $\tilde{\interp}$ such that $\tilde{k} \leq k^*$ is the unique $k$ such that $B(k,1) \leq \tilde{B} < B(k+1,1)$, $\tilde{\interp}$ is the unique solution to $\tilde{B} = B(\tilde{k},\tilde{\interp})$, and $R(\tilde{k},\tilde{\interp}) = 1$. Further, $\tilde{\interp} \in \left[\frac{c_{\tilde{k}}}{c_{\tilde{k}+1}},1\right]$.
\end{claim}
\begin{proof}
By the fact that $R(k,\interp)$ is increasing in $\interp$ and by Claim \ref{clm: property_zeta}, we get that: $R(k^*+1,1) = R\left(k^*,\frac{c_{k^*}}{c_{k^*+1}}\right) \leq R(k^*,\interp^*) < 1$.
Because $R(0,1) = 2$ and $R(k,1)$ is decreasing in $k$, there exists $\tilde{k} \leq k^*$ such that $R(\tilde{k},1) > 1 \geq R(\tilde{k}+1,1) =  R\left(\tilde{k},\frac{c_{\tilde{k}}}{c_{\tilde{k}+1}}\right)$. Because $R(\tilde{k},\interp)$ is increasing continuous in $\interp$, there exists $\tilde{\interp} \in \left[\frac{c_{\tilde{k}}}{c_{\tilde{k}+1}},1\right]$ such that  $R(\tilde{k},\tilde{\interp}) = 1$. Let $\tilde{B} = B(\tilde{k},\tilde{\interp}) = Q(\tilde{k},\tilde{\interp})$, we have $B(\tilde{k},1)  \leq \tilde{B} \leq B\left(\tilde{k},\frac{c_{\tilde{k}}}{c_{\tilde{k}+1}}\right) = B(\tilde{k}+1,1)$ as $B$ is decreasing in $\interp$. Thus $(\tilde{B}, \tilde{k}, \tilde{\interp})$ satisfies the claim.
\end{proof}

Now we show the proof of the lemma:
\begin{proof}[Proof of Lemma~\ref{lem: case2}]
Let us define $(\tilde{B}, \tilde{k}, \tilde{\interp})$ as in the statement of Claim \ref{clm:tilde-B}. First we show that $t^* \geq \tilde{k}$. This follows from noting that $\bar{B} > Q(k^*,\interp^*) \geq Q(k^*,1) \geq Q(\tilde{k},1)$ as clearly $Q(k,\interp)$ is increasing in $k$ and decreasing in $\interp$. Now, let $\duals_t = \frac{c_t}{c_{\tilde{k}}} \tilde{\interp}$ for $t \leq \tilde{k}$, $\duals_{\tilde{k}+1} = \ldots = \duals_{\ubc} = 1$. 

\emph{We prove that $\duals$ is a best response to $A$.} We have that
$$
2 \ldot{\pi}{\duals} = 2 \left( \frac{\tilde{\interp}}{c_{\tilde{k}}} \sum_{t=1}^{\tilde{k}} \pi_t c_t  + \sum_{t=\tilde{k}+1}^{\ubc} \pi_t \right) = R(\tilde{k},\tilde{\interp}) = 1
$$
For $j \leq \tilde{k} \leq t^*$, $A_j = 1$ hence $2 \ldot{\pi}{\duals} = \frac{1}{A_j}$. For $j > t^*$, we note that by definition of $t^*$, 
\begin{align*}
\bar{B}\leq Q(t^*+1, 1) = \sum_{t=1}^{t^*+1} \pi_t c_t + \sum_{t=t^*+2}^{\ubc} \pi_t \sqrt{c_t\cdot c_{t^*+1}} 
& = \sum_{t=1}^{t^*} \pi_t c_t + \sqrt{c_{t^*+1}}\sum_{t=t^*+1}^{\ubc} \pi_t \sqrt{c_t}
\\& \Rightarrow \frac{\bar{B}- \sum_{t=1}^{t^*} \pi_t c_t}{\sqrt{c_{t^*+1}}\sum_{t=t^*+1}^{\ubc} \pi_t \sqrt{c_t}} \leq 1
\end{align*} 
and this directly implies that $A_{t^*+1} \leq 1$. By monotonicity of $c_t$, it then holds that $A_j\leq 1$ for any $j> t^*$. Hence $2 \ldot{\pi}{\duals} = 1 \geq \frac{1}{A_j}$. This proves $\duals$ is a best response to $A$ by Lemma~\ref{lem: best-resp-zeta}, and simultaneously proves monotonicity of $A$.

\emph{Finally, we prove that $A$ is a best response to $\duals$.} By Lemma~\ref{lem: best-resp-A} it suffices to show that  $\min\left(1,\sqrt{\frac{\duals_t}{\lambda^* c_t}}\right) = A_t$ for $\lambda^*$ that makes the budget constraint binding. First, we note that the $A$ given by the lemma makes the budget constraint tight, as any allocation of the form given by Equation~ \eqref{eqn:generic-form-appendix} does so by construction. Now, let $\lambda^* = \left( \frac{\sum_{t=t^*+1}^{\ubc} \pi_t c_t}{\bar{B} - \sum_{t=1}^{t^*} \pi_t c_t} \right)^2$, the following holds:
\begin{enumerate}
\item For $j>t^*$:
$$
\sqrt{\frac{\duals_j}{\lambda^* c_j}} = \frac{1}{\sqrt{c_j}} \cdot \frac{\bar{B} - \sum_{t=1}^{t^*} \pi_{t} c_{t}}{\sum_{t=t^*+1}^{\ubc} \pi_{t} \sqrt{c_{t}}} = A_j \leq 1
$$
\item For $j\in \{\tilde{k}+1,\ldots, t^*\}$ (when this regime exists, i.e. when $\tilde{k} < t^*$), remember that by definition of $t^*$,
\begin{equation*}
\bar{B} > Q(t^*,1) =  \sum_{t=1}^{t^*} \pi_{t} c_{t} + \sqrt{c_{t^*}} \sum_{t=t^*+1}^{\ubc} \pi_{t} \sqrt{c_{t}} \Rightarrow \sqrt{\frac{1}{\lambda^* c_{t^*}}} = \frac{1}{\sqrt{c_{t^*}}}\frac{\bar{B} - \sum_{t=1}^{t^*} \pi_t c_t}{\sum_{t=t^*+1}^{\ubc} \pi_t \sqrt{c_t}}> 1
\end{equation*} 
Hence $\sqrt{\frac{1}{\lambda^* c_j}} \geq \sqrt{\frac{1}{\lambda^* c_{t^*}}}>1$ and $\min\left(1,\sqrt{\frac{1}{\lambda^* c_j}}\right)=1=A_j$.
\item For $j \leq \tilde{k}$, by Claim \ref{clm:tilde-B} we have $\tilde{\interp} \geq \frac{c_{\tilde{k}}}{c_{\tilde{k}+1}}$. Hence:$
\sqrt{\frac{\duals_j}{\lambda^* c_j}} = \sqrt{\frac{\tilde{\interp}}{\lambda^* c_{\tilde{k}}}}
\geq \sqrt{\frac{1}{\lambda^* c_{\tilde{k}+1}}} > 1$, 
where the latter holds by the previous case. Hence, $\min\left(1,\sqrt{\frac{1}{\lambda^* c_j}}\right)=1=A_j$.
\end{enumerate}
\end{proof}

\begin{lemma}[Case 3: $\bar{B} < B(1,1) = \frac{\sqrt{c_1} \E[\sqrt{c}]}{2}$] 
For this case, an optimal solution is given by $t^*=0$.
\end{lemma}

\begin{proof}
We let $\duals_t = 1$ for all $t$. We first show that $\duals_t$ is a best response to $A$. This trivially follows by Lemma~\ref{lem: best-resp-zeta} by remarking that 
$$
2 \ldot{\pi}{\duals} = 2 < \frac{\sqrt{c_1} \E[\sqrt{c}]}{\bar{B}} \leq \frac{\sqrt{c_j} \E[\sqrt{c}]}{\bar{B}} = \frac{1}{A_j}
$$
Now, we show that $A$ is a monotone best response to $\duals$. Monotonicity directly follows from the fact that the costs are non-decreasing. Now, to check that $A$ is a best response let us set $\lambda^* = \left(\frac{\E[\sqrt{c}]}{B}\right)^2$. We have $\sqrt{\frac{\duals_j}{\lambda^* c_j}} = \sqrt{\frac{1}{c_j}}\frac{\bar{B}}{\E[\sqrt{c}]} = A_j \leq \frac{\bar{B}}{\sqrt{c_1} \E[\sqrt{c}]}\leq 1/2$. Hence $A_j = \min\left(1,\sqrt{\frac{\duals_j}{\lambda^* c_j}}\right)$. 
By Lemma~\ref{lem: best-resp-A}, $A$ is a best response to $\duals$. This concludes the proof.
\end{proof}

\subsection{Proof of Theorem \ref{thm:main-moment-estimation-continuous}}\label{app:moments.continuum}
Consider any continuous atomless distribution $F$ supported in $[0,1]$. Then we can approximate the density of any such distribution by considering a discretized $\epsilon$-grid of the interval, 
i.e. $\{0,\epsilon, 2\epsilon, \ldots, 1\}$ and the discrete support distribution defined by pdf $\pi_t = F(\epsilon\cdot (t+1)) - F(\epsilon \cdot t)$ for $t\in \{0, 1, \ldots, 1/\epsilon\}$. Since, the loss of the zero sum game for moment estimation is continuous in the CDF of the cost distribution, we have that the minimax value of the game is continuous in the CDF of the cost distribution (see e.g. \cite{Feinberg2016} on continuity of minimax with respect to parameters of the game). Hence, the limit of the optimal discretized solutions will be the optimal solution to the discrete problem.

We now consider the limit structure of the optimal solutions to the discretized problems. We will use the more structural characterization of our main Theorem~\ref{thm:main-moment-estimation}, presented as Theorem~\ref{thm:main-moment-estimation-closed} in Appendix~\ref{app:proof-of-moment-estimation}.
In particular, the optimal solution will look as follows, taking the limit of the form in Theorem~\ref{thm:main-moment-estimation-closed}: for $x\in [0,1]$
\begin{equation}
A(x) = \begin{cases}
\bar{A} & \text{if $x\leq x^*$}\\
\frac{1}{\sqrt{x}}\cdot \frac{\bar{B} - \bar{A} \E[c\cdot 1\{c\leq x^*\}]}{\E[\sqrt{c}\cdot 1\{c>x^*\}} & \text{o.w.}
\end{cases}
\end{equation}
Now let us examine how the point $x^*$ is defined in the limit. Consider the functions $Q(k,z)$, $R(k,z)$ and $B(k,z)$ defined in Theorem \ref{thm:main-moment-estimation-closed}. In the limit as $\epsilon\rightarrow 0$, observe that $c_k\rightarrow c_{k+1}$ for every $t$. Therefore, it is easy to see that $Q(k, 1)\rightarrow Q(k+1,1)$ and $R(k,1)\rightarrow R(k+1,1)$. Hence, also $B(k,1)\rightarrow B(k+1,1)$. Hence, we will have that $x^*$ defined in \ref{thm:main-moment-estimation-closed} will satisfy $x^*\rightarrow 1$ and in the limit $\bar{B}=B(k^*,1)$ for some $k^*\in [0, \infty)$. Hence, we only need to consider these functions at $x=1$ and take their limit as $\epsilon\rightarrow 0$. In this limit we observe that these functions take the simpler forms (since summations will converge to integrals) for $x\in [0,1]$
\begin{align}
Q(x/\epsilon, 1) \rightarrow~& \int_{0}^{x} c f(c) dc+ \int_{x}^{1} \sqrt{c\cdot x} f(c)  dc = \E_{c \sim \cF}[\min\{c, \sqrt{cx}\}] \triangleq  Q_{\infty}(x)\\
R(x/\epsilon, 1) \rightarrow~& 2\left(\int_{0}^{x} f(c) \frac{c}{x}dc + \int_{x}^{1}f(c)dc \right) = 2\E_{c \sim \cF}\left[\min\left\{\frac{c}{x}, 1\right\}\right] \triangleq R_{\infty}(x)
\end{align}
Hence, adapting the discrete characterization of $x^*$ and $\bar{A}$ to these limits we have: the parameter $x^*$ is defined as the solution to the following process: let $z$ be the solution to the equation $\bar{B} = \frac{Q_{\infty}(z)}{R_{\infty}(z)}$. If $R_{\infty}(z)\geq 1$, then $x^*=z$ and $\bar{A} = \frac{1}{R_{\infty}(z)}$, otherwise $x^*=\max\{z: \bar{B} > Q_{\infty}(z)\}$ and $\bar{A}=1$.

Now we observe that since $F$ is atomless and has support $[0,1]$, $Q_{\infty}(x)$ is a continuous increasing function of $x$, with range $[0,\E[c]]$. Hence, if $\bar{B}\leq \E[c]$ (which we assumed holds as otherwise the problem is trivial), then $\max\{z: \bar{B} > Q_{\infty}(z)\} = Q_{\infty}^{-1}(\bar{B})$, or equivalently if $x^*=\max\{z: \bar{B} > Q_{\infty}(z)\}$ then it must be that $x^*$ is the unique solution to the equation $Q_{\infty}(x^*)=\bar{B}$. 

Moreover, observe that $R_{\infty}(x)$ is also a decreasing function of $x$ ranging in $[2\E[c],2]$ as $x$ varies from $0$ to $1$. If $2\E[c] \geq 1$, then $R_{\infty}(x)\geq 1$ for all $x\in [0,1]$ and the second case of the characterization of $x^*$ never holds and we have that $x^*$ is the solution to the equation $\frac{Q_{\infty}(x^*)}{R_{\infty}(x^*)}=\bar{B}$, or $1$ if $\bar{B}$ is above $\frac{1}{2}$. Moreover, $\bar{A}=\frac{1}{R_{\infty}(x^*)}$, in both cases. Equivalently, $x^*=\min\{1, G^{-1}(\bar{B})\}$. Hence, in this case the Theorem holds.

Otherwise, let $x_0$ be the solution to the equation $R_{\infty}(x_0)=1$. Thus $R_{\infty}(x)<1$ above $x_0$ and $R_{\infty}(x)>1$ below $x_0$. Now, consider the function $G(x) = \frac{Q_{\infty}(x)}{\max(1,R_{\infty}(x))}$. This function is
continuous increasing and is equal to $\frac{Q_{\infty}(x)}{R_{\infty}(x)}$ for $x\leq x_0$ and is equal to $Q_{\infty}(x)$ for $x\geq x_0$. 

If it happened that the solution of the equation $\frac{Q_{\infty}(x)}{R_{\infty}(x)}=\bar{B}$ happens at $x\leq x_0$, then we have that $x^* = G^{-1}(\bar{B})$ and $\bar{A}=\frac{1}{R(x^*)}$. Otherwise, if the solution to that equation is above $x_0$, then $x^*=Q^{-1}(\bar{B})=G^{-1}(\bar{B})$ (the latter always has a solution when $\bar{B}\leq \E[c]$) and $\bar{A} = 1$. Thus in this case we have that $x^* = G^{-1}(\bar{B})$ and $\bar{A}= \frac{1}{\max\{1,R_{\infty}(x^*)\}}$, which concludes the proof.

\subsection{Proof of Lemma \ref{lem:ht-variance}}\label{app:ht-variance}
For any distribution $\cD$, observe that each coordinate of the Horvitz-Thompson estimator can be written as the sum of $n$ i.i.d. random variables each with a variance:
\begin{align*}
\sigma_r^2 = \E\left[\left(\frac{m_r(z_i)\cdot 1\{i\in S\}}{A(c_i)}\right)^2\right] - \E\left[\frac{m_r(z_i)\cdot 1\{i\in S\}}{A(c_i)}\right]^2 = \E\left[\frac{\E\left[m_r(z)^2|c\right]}{A(c)}\right] - \E[m_r(z)]^2
\end{align*}
Hence, the variance of the estimator is $\sum_{r=1}^d\nicefrac{\sigma_r^2}{n}$. For simplicity of notation we let $M$ denote the random value of the moment vector, i.e. the random variable generated by first drawing a $(z,c)\sim \cD$ and then mapping $z$ to $m(z)$. Then we can write:
\begin{align*}
n\cdot \Risk(\hat{\theta}_S; \cD, A) =~ \sum_{r=1}^d \left(\E\left[\frac{\E\left[M_r^2|c\right]}{A(c)}\right]  - \E[M_r]^2\right)
\end{align*}

We begin by showing that the worst-case risk is achieved by a distribution of $M$ supported only on the corners of the hypercube, i.e. $\{0,1\}^d$. Since any such distribution is a valid one by the assumption that the distribution of $m(z)$ is supported on all corners of the hypercube, it is without loss of generality to restrict attention to only such distributions when characterizing the worst-case risk. 

Let $Y$ be a new random vector that is generated as follows: conditional on any value of $M\in [0,1]^d$, we draw a random threshold $\tau$ uniformly in $[0,1]$. Then we define $Y_r = 1\{M_r \geq \tau\}$. Observe that $Y$ is a mean-preserving transformation of $M$, i.e. the expected value of $Y_r$ is equal to the expected value of $M_r$:
\begin{equation*}
\E[Y_r] = \E[ \E[Y_r|M_r] ] = \E[ \E[1\{M_r \geq \tau\}|M_r] ] = \E[M_r]
\end{equation*}
More over observe that $Y$ achieves a strictly larger risk:
\begin{align*}
\E\left[\frac{\E\left[M_r^2|c\right]}{A(c)}\right] - \E[M_r]^2 =~&
\E\left[\frac{\E\left[\E[Y_r|M_r]^2|c\right]}{A(c)}\right] - \E[Y_r]^2\\
\leq~& \E\left[\frac{\E\left[\E[Y_r^2|M_r]|c\right]}{A(c)}\right] - \E[Y_r]^2 \tag{By Jensen's inequality}\\
=~& \E\left[\frac{\E\left[Y_r^2|c\right]}{A(c)}\right]- \E[Y_r]^2
\end{align*}
Thus when taking a supremum of the risk over consistent distributions $\cD$, we can assume that the moment vector $M$ under such distributions is only supported on the corners of the hypercube $\{0,1\}^d$.

For any such distribution observe that $M_r^2 = M_r$ and $\E[M_r^2|c]=\E[M_r|c]=\Pr[M_r=1|c]$. For simplicity, let $\duals_{tr}=\Pr[M_r=1|c_t]$. Then we can write the risk of any such distribution of $M_r$ as:
\begin{align*}
n\cdot \Risk(\hat{\theta}_S; \cD, A) 
&=~ \sum_{r=1}^d \left(\E\left[\frac{\Pr[M_r=1|c]}{A(c)}\right]- \E\left[\Pr[M_r=1|c]\right]^2\right)
\\&=~ \sum_{r=1}^d \left(\sum_{t=1}^{\ubc} \pi_t  \cdot \frac{\duals_{tr}}{A_t} - \left(\sum_{t=1}^{\ubc} \pi_t \cdot \duals_{tr}\right)^2\right)
\end{align*}

Finally, observe that conditional on any cost $c_t$, any vector $\duals_t\in [0,1]^d$ is a feasible vector (i.e. can be implemented by a consistent distribution $\cD$). To see this, let $\duals_t$ be any such vector. Then we generate a distribution that leads to this $q_t$ as follows: conditional on the cost $c_t$, we draw a random threshold $\tau$ uniformly in $[0,1]$. Then set $M_r = 1\{\duals_{tr}\geq \tau\}$. The latter is a valid distribution for $M_r$ and it satisfies that $\Pr[M_r=1|c_t]=\duals_t$. 

Thus when calculating the worst-case risk, we are taking a supremum over all vectors $q\in [0,1]^{d\cdot \ubc}$, i.e.:
\begin{align*}
n\cdot \Risk^*(\hat{\theta}_S; \cF, A) = \sup_{\duals \in [0,1]^{d\cdot \ubc}} \sum_{r=1}^d\left(
\sum_{t=1}^{\ubc} \pi_t \cdot \frac{\duals_{tr}}{A_t} - \left(\sum_{t=1}^{\ubc} \pi_t\cdot \duals_{tr}\right)^2\right)
\end{align*}
Since the latter optimization has no constraints to couple the objectives inside the first summation, the latter boils down to optimizing each summand separately, i.e.
\begin{align*}
\sup_{\duals_r \in [0,1]^{\ubc}} 
\sum_{t=1}^{\ubc} \pi_t \cdot \frac{\duals_{tr}}{A_t} - \left(\sum_{t=1}^{\ubc} \pi_t\cdot \duals_{tr}\right)^2
\end{align*}
Since the objective is identical for all $r\in \{1,\ldots, d\}$, we get that there always exists an optimal solution under which $q_r = q_{r'}$ for all $r, r'\in \{1,\ldots, d\}$. Thus we can re-write the worst-case risk by restricting only to such symmetric solutions for $q$, i.e.:
\begin{align*}
n\cdot \Risk^*(\hat{\theta}_S; \cF, A) 
& = \sup_{\duals \in [0,1]^{\ubc}} \sum_{r=1}^d\left(
\sum_{t=1}^{\ubc} \pi_t \cdot \frac{\duals_{t}}{A_t} - \left(\sum_{t=1}^{\ubc} \pi_t\cdot \duals_{t}\right)^2\right)
\\& = \sup_{\duals \in [0,1]^{\ubc}} d\cdot \left(
\sum_{t=1}^{\ubc} \pi_t \cdot \frac{\duals_{t}}{A_t} - \left(\sum_{t=1}^{\ubc} \pi_t\cdot \duals_{t}\right)^2\right)
\end{align*}

\subsection{Proof of Theorem \ref{thm:main-regression}}\label{app: regression}
Before starting to prove the Theorem, we note that we can assume without loss of generality that $U > 0$. Otherwise, $U = L = 0$ (as we are assuming $U^2 \geq L^2$), and the objective value of the optimization program is $0$ and independent of the allocation rule, thus any feasible allocation rule is optimal. In particular, any monotone allocation rule of the form given in the theorem statement works. 

\paragraph{Simplifying the analyst's problem} The following lemma reduces the minimax problem that the analyst needs to solve to a simple convex minimization problem:
\begin{lemma}
The optimization program solved by the analyst can be written as:
\begin{equation}\label{eq: app-regression-program}
\begin{aligned}
\inf_{A \in \left(0,1\right]^{\ubc}} ~&\sum_{t=1}^{t^*-1}  \pi_t \frac{L^2}{A_t} +  \pi_{t^*} \frac{R^2}{A_{t^*}} + \sum_{t = t^* +1}^{\ubc} \pi_t \frac{U^2 }{A_t} \\
\mathrm{s.t.}~~~
~& \sum_{t=1}^{\ubc} \pi_t c_t A_t \leq \bar{B}\\
~& A \text{ is monotone non-increasing }\\
\end{aligned}
\end{equation}
where $t^* = \min \{j:~- \frac{L}{U - L}  > \sum_{t=j+1}^{\ubc} \pi_t\}$, $q_t^* = \frac{1}{\pi_t^*} \left( = - \frac{L}{U - L}  - \sum_{t=t^*+1}^{\ubc} \pi_t \right) > 0$ and $R^2 = (U^2- L^2) q_{t^*} + L^2 > 0$.
\end{lemma}

\begin{proof}
We first note that for a given $A$, the optimization program solved by the adversary can be rewritten as
\begin{equation*}
\begin{aligned}
\sup_{\duals \in [0,1]^{\ubc}}  &\sum_{t=1}^{\ubc} \frac{\pi_t}{A_t} \duals_t (U^2 -  L^2) + L^2 \sum_{t=1}^{\ubc} \frac{\pi_t}{A_t}\\
\mathrm{s.t.}~~&\sum_t \pi_t \duals_t = -\frac{L}{U-L}
\end{aligned}
\end{equation*}
Hence, the adversary is exactly solving a knapsack problem with capacity $- \frac{L}{U - L} \geq 0$, weights $\pi_t$ and utilities $\frac{\pi_t}{A_t}$. Therefore, an optimal solution exists and is to put the weight on the $t$'s with the higher values of $\frac{1}{A_t}$ i.e. the lowest values of $A_t$ first. Because $A$ is non-increasing in the costs and therefore in $t$, an optimal solution is given by:
\begin{align*}
q_1 = \ldots = q_{t^*-1} = 0, \pi_{t^*}q_{t^*} = - \frac{L}{U - L}  - \sum_{t=t^*+1}^{\ubc} \pi_t, q_{t^*+1} = \ldots = q_l = 1
\end{align*}
This holds independently of $A$  as long as $A$ is feasible (hence monotone), proving the result.
\end{proof}

\paragraph{Solving the optimization problem} We start with the following lemma that characterizes the form of the solution: 

\begin{lemma}\label{lem: gmm-form-solution}
There exist $\lambda \geq 0$, non-negative integers $t^-,t^+$ such that $t^- \leq t^* \leq t^+$, and an optimal allocation rule $A$ that satisfy
\begin{enumerate}
\item $A_1  \geq \ldots \geq A_{t^- -1}  \geq A_{t^-} = \ldots = A_{t^+} \geq A_{t^+ + 1} \geq \ldots A_l$
\item $A_t = \min\left(1,\sqrt{\frac{L^2}{\lambda c_t}}\right) \; \forall t < t^-$
\item $A_t = \min\left(1,\sqrt{\frac{ U^2}{\lambda c_t}}\right) \; \forall t > t^+$
\end{enumerate}
and that make the budget constraint tight. In the rest of the proof, we denote $\bar{A}  \triangleq A_{t^-} = \ldots = A_{t^+}$.
\end{lemma}

\begin{proof}
First, we show existence of an optimal solution $A$. Because the optimal value of the program is finite but the objective tends to infinity when any $A_{\ubc}$ tends to $0$ (as we have $U^2,R^2 > 0$), there must exists $\gamma > 0$ such that the analyst's program is given by
\begin{equation}
\begin{aligned}
\inf_{A \in [\gamma,1]^{\ubc}} ~&\sum_{t=1}^{t^*-1}  \pi_t \frac{L^2}{A_t} +  \pi_{t^*} \frac{R^2}{A_{t^*}} + \sum_{t = t^* +1}^{\ubc} \pi_t \frac{U^2 }{A_t} \\
\mathrm{s.t.}~~~
~& \sum_{t=1}^{\ubc} \pi_t c_t A_t \leq \bar{B}\\
~& A \text{ is monotone non-increasing }\\
\end{aligned}
\end{equation}
The objective is convex and continuous in $A$ over $[\gamma,1]^{\ubc}$, and the feasible set is convex and compact, therefore the above program admits an optimal solution. Now, consider the following program with partial monotonicity constraints:
\begin{equation}
\begin{aligned}
\inf_{A \in \left(0,1\right]^{\ubc}} ~&\sum_{t=1}^{t^*-1}  \pi_t \frac{L^2}{A_t} +  \pi_{t^*} \frac{R^2}{A_{t^*}} + \sum_{t = t^* +1}^{\ubc} \pi_t \frac{U^2 }{A_t} \\
\mathrm{s.t.}~~~
~& \sum_{t=1}^{\ubc} \pi_t c_t A_t \leq \bar{B}\\
~& A_1, \ldots, A_{t^*-1} \geq A_t \geq A_{t*+1}, \ldots, A_l
\end{aligned}
\end{equation}

In fact, there exists an optimal solution to this problem that makes the budget constraint tight (one can increase the allocation rule without decreasing the objective value until the budget constraint becomes tight). Consider such a solution, the Lagrangian of the program is given by
\begin{align*}
\cL(A,\lambda,\lambda_t^1,\lambda_t) 
& =  \sum_{t=1}^{t^*-1}  \pi_t \frac{L^2}{A_t} +  \sum_{t = t^* +1}^{\ubc} \pi_t \frac{U^2 }{A_t} +  \pi_{t^*} \frac{R^2}{A_{t^*}} + \lambda \sum_{t} \pi_t c_t A_t 
\\& + \sum_{t} \lambda_t^1 A_t - \lambda B/n -\sum_t \lambda_t^1 + \sum_{t < t^*} \lambda_t  (A_{t^*} - A_t) + \sum_{t > t^*} \lambda_t  (A_t - A_{t^*})
\end{align*}
It must necessarily be the case at optimal that whenever $A_t < 1$, $\lambda_t^1 = 0$ and whenever $A_t > A_{t^*}$ (for $t < t^*$) or $A_t < A_{t^*}$ (for $t > t^*$), $\lambda_t = 0$ by the KKT conditions, and that $A$ satisfies the first order conditions $\nabla_A \cL(A,\lambda,\lambda_t^1,\lambda^k)  = 0$. Hence, an optimal solution must necessarily satisfy:
\begin{enumerate}
\item $A_t = \min\left(1,\max\left(A_{t^*},\sqrt{\frac{\pi_t L^2}{\lambda \pi_t c_t}}\right) \right)$ for $t < t^*$
\item $A_t = \min\left(1,A_{t^*},\sqrt{\frac{\pi_t U^2}{\lambda \pi_t c_t }}\right)$ for $t > t^*$
\end{enumerate}
Note that this implies that $A$ is monotone non-increasing (as the virtual costs are monotone non-decreasing), therefore is an optimal solution to the analyst's problem, and that there must exist $t^-$ and $t^+$ such that
$$
A_1 \geq \ldots \geq A_{t^- -1} \geq A_{t^-} = \ldots = A_{t^*} = \ldots = A_{t^+} \geq A_{t^+ + 1} \geq \ldots A_l
$$ 
with $A_t = \min\left(1,\sqrt{\frac{L^2}{\lambda c_t}}\right) \; \forall t < t^-$ and $A_t = \min\left(1,\sqrt{\frac{ U^2}{\lambda c_t}}\right) \; \forall t > t^+$.
\end{proof}

We now proceed onto proving the main statement. Let $\gamma_t = |L|$ (resp. $R$, $U$), for $t < t^*$ (resp. $t=t^*$, $t > t^*$). Let $t^-,t^+$ be such that $t^- \leq t^* \leq t^+$ and $A_{t^-} = \ldots = A_{t^+}$ at optimal. Suppose the analyst has knowledge of $t^-,t^+$. Then, replacing $A$ as a function of $t^-, t^+, \lambda,\bar{A}$ in the analyst's problem~\eqref{eq: app-regression-program} reduces to the following problem of two variables $\lambda$ and $\bar{A}$:
\begin{equation}\label{eq: 2var_GMM_problem}
\begin{aligned}
\inf_{\lambda \geq 0,~\bar{A} \in [0,1]} ~&\sum_{t = 1}^{t^--1} \pi_t   \frac{\gamma_t^2}{\min(1,\frac{\gamma_t}{\sqrt{\lambda c_t}})} + \frac{1}{\bar{A}} \cdot \sum_{t=t^-}^{t^+-1} \pi_t \gamma_t + \cdot \sum_{t = t^+}^{\ubc} \pi_t  \frac{\gamma_t^2}{\min(1,\frac{\gamma_t}{\sqrt{\lambda c_t}})} \\
\mathrm{s.t.}~~~
~& \sum_{t=1}^{t^- -1} \pi_t  c_t \min(1,\frac{\gamma_t}{\sqrt{\lambda c_t}}) + \bar{A} \sum_{t=t^-}^{t^+} \pi_t c_t  + \sum_{t=t^+ +1}^{\ubc} \pi_t c_t \min(1,\frac{\gamma_t}{\sqrt{\lambda c_t}}) = \bar{B}\\
~&\min\left(1,\frac{\gamma_t}{\sqrt{\lambda c_t}}\right) \geq \bar{A}~~\forall t \in \{t^1, \ldots, t^- - 1\} \\
~&\min\left(1,\frac{\gamma_t}{\sqrt{\lambda c_t}}\right) \leq \bar{A}~~\forall t \in \{t^+ + 1,\ldots, \ubc\}\\
\end{aligned}
\end{equation}
 Hence there exists an optimal solution of the form given by Lemma~\ref{lem: gmm-form-solution} with 
$$
\bar{A} = A(\lambda) \triangleq \frac{1}{\sum_{t=t^-}^{t^+} \pi_t c_t} \left(  \bar{B} - \sum_{t=1}^{t^- -1} \pi_t  c_t \min\left(1,\frac{\gamma_t}{\sqrt{\lambda c_t}}\right) - \sum_{t=t^+ +1}^{\ubc} \pi_t  c_t \min\left(1,\frac{\gamma_t}{\sqrt{\lambda c_t}}\right) \right)
$$
as $\bar{A}$ as a function of $\lambda$ is entirely determined by the budget constraint. Plugging this back in the above program, we can rewrite the program as depending only on the variable $\lambda$ as follows,
also remembering that the (virtual) costs are monotone non-decreasing:
\begin{equation}\label{eq:1var_GMM_problem}
\begin{aligned}
\inf_{\lambda \geq 0} ~&\sum_{t = 1}^{t^--1} \pi_t  \frac{\gamma_t^2}{\min\left(1,\frac{\gamma_t}{\sqrt{\lambda c_t}}\right)}  + \frac{1}{A(\lambda)} \cdot \sum_{t=t^-}^{t^+} \pi_t \gamma_t + \cdot \sum_{t = t^+ +1}^{\ubc} \pi_t \frac{\gamma_t^2}{\min\left(1,\frac{\gamma_t}{\sqrt{\lambda c_t}}\right)}\\
\mathrm{s.t.}~~~
~&\min\left(1,\frac{|L|}{\sqrt{\lambda c_{t^--1}}}\right) \geq A(\lambda) \\
~&\min\left(1,\frac{U}{\sqrt{\lambda c_{t+ +1}}}\right) \leq A(\lambda) \\
\end{aligned}
\end{equation}

Suppose that the analyst also has knowledge of $t^1$, the maximum value of $t \in \{0,\ldots,t^--1\} \cup \{t^++1, \ldots, \ubc \}$ such that $\min(1,\frac{\gamma_t}{\sqrt{\lambda c_{t}}}) = 1$. Then $\lambda^*$ must satisfy one of the three following conditions:
\begin{enumerate}
\item $\min\left(1,\frac{|L|}{\sqrt{\lambda c_{t^--1}}}\right) = A(\lambda)$. Letting $\mu = \frac{1}{\sqrt{\lambda}}$, this is a linear equation in $\mu$ and therefore can be solved efficiently given knowledge of $t^1$. This follows from noting the equation can be written
$$
\min\left(1,\frac{\mu |L|}{\sqrt{c_{t^--1}}}\right) 
=  \frac{1}{\sum_{t=t^-}^{t^+} \pi_t c_t} \left(  \bar{B} - \sum_{t=1}^{t^- -1} \pi_t  c_t \min\left(1,\frac{\mu \gamma_t}{\sqrt{c_t}}\right) - \sum_{t=t^+ +1}^{\ubc} \pi_t  c_t \min\left(1,\frac{\mu \gamma_t}{\sqrt{c_t}}\right) \right)
$$
which is of the form  $a \mu + b = 0$ for some constants $a,b$.
\item $\min(1,\frac{U}{\sqrt{\lambda c_{t+ +1}}}) = A(\lambda)$.  Letting $\mu = \frac{1}{\sqrt{\lambda}}$, this is a linear equation in $\mu$ and can be solved efficiently given knowledge of $t^1$. This follows from the same argument as above.
\item $\lambda$ minimizes the optimization problem with only the non-negativity constraint $\lambda \geq 0$. Then, letting $\mu = \frac{1}{\sqrt{\lambda}}$, the objective value is the following function of $\mu$:
\begin{align*}
OPT(\mu) &= \sum_{t = 1}^{t^--1} \pi_t   \gamma_t^2 \max\left(1,\frac{\sqrt{c_t}}{\mu \gamma_t} \right) + \sum_{t = t^+}^{\ubc} \pi_t   \gamma_t^2 \max\left(1,\frac{\sqrt{c_t}}{\mu \gamma_t}  \right) 
\\&+ \frac{ \sum_{t=t^-}^{t^+} \pi_t c_t \cdot \sum_{t=t^-}^{t^+-1} \pi_t \gamma_t }{  \bar{B} - \sum_{t=1}^{t^- -1} \pi_t  c_t \min\left(1,\frac{\mu \gamma_t}{\sqrt{c_t}}\right) - \sum_{t=t^+ +1}^{\ubc} \pi_t  c_t \min\left(1,\frac{\mu \gamma_t}{\sqrt{c_t}}\right) }
\end{align*}
i.e. can be written (with knowledge of $t^1$) $OPT(\mu) = C + \frac{K}{\mu} + \frac{1}{\gamma - \kappa \mu}$ for some constants $C,\gamma, K, \kappa$. The first order condition is given by $\frac{K}{\mu^2} = \frac{\kappa}{(\gamma - \kappa \mu)^2}$ and a minimizer can therefore be computed efficiently (either as $0$, $+\infty$, or as a solution of the first order conditions---whichever leads to the smallest objective value). 
\end{enumerate}
The analyst only needs to pick the value of $\lambda$ among the three cases above that is feasible and minimizes the objective value of Program~\eqref{eq:1var_GMM_problem}. In practice, the analyst does not know $t^1$ but can search over the space of possible $t^1$'s, which can be done in at most $\ubc$ steps (note that there may be values of $t^1$ for which the program is infeasible, showing that said value of $t^1$ is impossible). The analyst can solve the optimization problem absent knowledge of $t^-$ and $t^+$ by searching over $(t^-,t^+)$ pairs, and the analyst obtains a solution by picking the $(t^1,t^-,t^+)$ tuple for which Program~\eqref{eq:1var_GMM_problem} is feasible and the corresponding optimal $\lambda$ that lead to the best objective value over all tuples. This can be done in a most $\ubc^3$ steps.
\end{appendix}

\end{document}